\documentclass[10pt]{article}

\usepackage{times}
\usepackage[numbers,sort]{natbib}
\usepackage{amsmath,latexsym,amsthm,amssymb,amscd,url,enumerate,mathtools}
\usepackage[show]{ed}
\usepackage[all]{xy}
\usepackage{tikz,float}
\usepackage{multirow}
\usepackage{relsize}
\usepackage{caption}
\usepackage{graphicx}
\usepackage{graphics}
\usepackage{leftidx}
\usepackage{xcolor}

\begin{document} 
\author{Kristina Sojakova and Patricia Johann}

\title{A General Framework for Relational Parametricity}
\date{January 2018}

\maketitle
\newcommand{\Ccal}{\C}
\newcommand{\R}{\mathcal{R}}
\newcommand{\C}{\mathcal{C}}
\newcommand{\X}{\mathcal{X}}
\newcommand{\Y}{\mathcal{Y}}
\newcommand{\Z}{\mathcal{Z}}
\newcommand{\F}{\mathcal{F}}
\newcommand{\G}{\mathcal{G}}
\newcommand{\Gcal}{\G}
\newcommand{\E}{\mathcal{E}}
\newcommand{\B}{\mathcal{B}}
\newcommand{\Ucal}{\U}
\newcommand{\U}{\mathbb{U}}
\newcommand{\M}{\mathcal{M}}
\newcommand{\I}{\mathcal{I}}
\newcommand{\termobj}{\mathbf{1}}
\newcommand{\source}{\mathsf{s}}
\newcommand{\target}{\mathsf{t}}
\newcommand{\fm}{\mathbf{f}}
\newcommand{\dm}{\mathbf{d}}
\newcommand{\cm}{\mathbf{C}}
\newcommand{\nat}{\mathbb{N}}
\newcommand{\two}{\mathbf{Bool}}
\newcommand{\fst}{\mathsf{fst}}
\newcommand{\snd}{\mathsf{snd}}
\newcommand{\eval}{\mathsf{eval}}
\newcommand{\cat}{\mathsf{Cat}}
\newcommand{\rel}{\mathcal{R}}
\newcommand{\set}{\mathsf{Set}}
\newcommand{\inj}{\mathsf{i}}
\newcommand{\Ob}{\mathbf{Ob}}
\newcommand{\Mor}{\mathbf{Mor}}
\newcommand{\ctx}{\mathsf{Ctx}}
\newcommand{\Set}{\mathsf{Set}}
\newcommand{\Prop}{\mathsf{Prop}}
\newcommand{\sem}[1]{\ensuremath{[\![ #1 ]\!]}}
\newcommand{\grph}[1]{\ensuremath{\langle #1 \rangle}}
\newcommand{\cart}[1]{\ensuremath{#1^{\S}}}
\newcommand{\opcart}[1]{\ensuremath{#1_{\S}}}
\newcommand{\foldd}[3][]{\ensuremath{\mathit{fold}_{#1}[#2, #3]}}
\newcommand{\inn}[1][]{\ensuremath{\mathit{in}_{#1}}}
\newcommand{\id}{\mathsf{id}}
\newcommand{\Eq}{\mathsf{Eq}}
\newcommand{\eq}{\mathsf{eq}}
\newcommand{\expo}[2]{\ensuremath{#1 \Rightarrow #2}}
\newcommand{\internal}[1]{\ensuremath{\underline{#1}}}
\newcommand{\Id}{\mathsf{Id}}

\newtheorem{theorem}{Theorem}						
\newtheorem{definition}[theorem]{Definition}
\newtheorem{proposition}[theorem]{Proposition}
\newtheorem{lemma}[theorem]{Lemma}
\newtheorem{remark}[theorem]{Remark}
\newtheorem{example}[theorem]{Example}
\newtheorem{notation}[theorem]{Notation}
\newtheorem{corollary}[theorem]{Corollary}


\begin{abstract}
Reynolds' original theory of {\em relational parametricity} was
intended to capture the idea that polymorphically typed System F
programs preserve all relations between inputs. But as Reynolds
himself later showed, his theory can only be formalized in a
meta-theory with an impredicative universe, such as the Calculus of
Inductive Constructions. Abstracting from Reynolds' ideas, Dunphy and
Reddy developed their well-known framework for parametricity that
uses parametric limits in reflexive graph categories and aims to
subsume a variety of parametric models. As we observe, however, their
theory is not sufficiently general to subsume the very model that
inspired parametricity, namely Reynolds' original model, expressed inside type theory.

To correct this, we develop an abstract framework for relational
parametricity that generalizes the notion of a reflexive graph
categories and delivers Reynolds' model as a direct instance in a
natural way. This framework is uniform with respect to a choice of
meta-theory, which allows us to obtain the well-known PER model of
Longo and Moggi as a direct instance in a natural way as well. In
addition, we offer two novel relationally parametric models of System
F: \emph{i)} \emph{a categorical version of Reynolds' model}, where
types are functorial on isomorphisms and all polymorphic functions
respect the functorial action, and \emph{ii)} \emph{a proof-relevant
  categorical version of Reynolds' model} (after Orsanigo), where,
additionally, witnesses of relatedness are themselves suitably
related. We show that, unlike previously existing frameworks for
parametricity, ours recognizes both of these new models in a natural
way. Our framework is thus \emph{descriptive}, in that it accounts for
well-known models, as well as \emph{prescriptive}, in that it
identifies abstract properties that good models of relational
parametricity should satisfy and suggests new constructions of such
models.
\end{abstract}

\section{Introduction}

Reynolds~\cite{param_reynolds} introduced the notion of {\em
  relational parametricity} to model the extensional behavior of
programs in System F~\cite{systemf}, the formal calculus at the core
of all polymorphic functional languages. His goal was to give a type
$\alpha \vdash T(\alpha)$ an \emph{object interpretation} $T_0$ and a
\emph{relational interpretation} $T_1$, where $T_0$ takes sets to sets
and $T_1$ takes relations $R \subseteq A \times B$ to relations
$T_1(R) \subseteq T_0(A) \times T_0(B)$. A term $\alpha; x : S(\alpha)
\vdash t(\alpha,x) : T(\alpha)$ was to be interpreted as a map $t_0$
associating to each set $A$ a function $t_0(A) : S_0(A) \to
T_0(A)$. The interpretations were to be given inductively on the
structure of $T$ and $t$ in such a way that they implied two key
theorems: the \emph{Identity Extension Lemma}, stating that if $R$ is
the equality relation on $A$ then $T_1(R)$ is the equality relation on
$T_0(A)$; and the \emph{Abstraction Theorem}, stating that, for any
relation $R \subseteq A \times B$, $t_0(A)$ and $t_0(B)$ map arguments
related by $S_1(R)$ to results related by $T_1(R)$. A similar result
holds for types and terms with any number of free variables.

In Reynolds' treatment of relational parametricity, if $U(\alpha)$ is
the type $\alpha \vdash S(\alpha) \to T(\alpha)$, for example, then
$U_0(A)$ is the set of functions $f : S_0(A) \to T_0(A)$ and, for $R
\subseteq A \times B$, $U_1(R)$ relates $f : S_0(A) \to T_0(A)$ to $g
: S_0(B) \to T_0(B)$ iff $f$ and $g$ map arguments related by $S_r(R)$
to results related by $T_1(R)$. Similarly, if $V$ is the type $\cdot
\vdash \forall \alpha. S(\alpha)$, then $V_0$ consists of those
polymorphic functions $f$ that take a set $A$ and return an element of
$S_0(A)$, and also have the property that for any relation $R
\subseteq A \times B$, $f(A)$ and $f(B)$ are related by $S_1(R)$. Two
such polymorphic functions $f$ and $g$ are then related by $V_1$ iff
for any relation $R \subseteq A \times B$, $f(A)$ and $g(B)$ are
related by $S_1(R)$. These definitions allow us to deduce interesting
properties of (interpretations of) terms solely from their types. For
example, for any term $t : \forall \alpha. \alpha \to \alpha$, the
Abstraction Theorem guarantees that the interpretation $t_0$ of $t$ is
related to itself by the relational interpretation of $\forall
\alpha. \alpha \to \alpha$. So if we fix a set $A$, fix $a\in A$, and
define a relation on $A$ by $R \coloneqq \{(a,a)\}$, then $t_0(A)$
must be related to itself by the relational interpretation of $\alpha
\vdash \alpha \to \alpha$ applied to $R$. This means that $t_0(A)$
must carry arguments related by $R$ to results related by $R$. Since
$a$ is related to itself by $R$, $t_0(A) \,a$ must be related to
itself by $R$, so that $t_0(A) \, a$ must be $a$. That is, $t_0$ must
be the polymorphic identity function. Such applications of relational
parametricity are useful in many different scenarios, {\em e.g.}, when
proving invariance of polymorphic functions under changes of data
representation, equivalences of programs, and ``free
theorems''~\cite{wad87}.

The well-known problem with Reynolds' treatment of relational
parametricity (see \cite{not-set-theoretic}) is that the universe of
sets is not impredicative, and hence the aforementioned ``set'' $V_0$
cannot be formed. This issue can be resolved if we instead work in a
meta-theory that has an impredicative universe; a natural choice is an
extensional version of the Calculus of Inductive Constructions (CIC),
{\em i.e.}, a dependent type theory with a cumulative Russell-style
hierarchy of universes $\Ucal_0 : \Ucal_1 : \ldots$, where $\Ucal_0$
is impredicative, and extensional identity types. With this
adjustment, we have two
canonical relationally parametric models of System F: \emph{i)} the
PER model of Longo and Moggi~\cite{longo_moggi}, internal to the
theory of $\omega$-sets and realizable functions, and \emph{ii)}
Reynolds' original model\footnote{Since there are no set-theoretic
  models of System F, by the phrase ``Reynolds' original model'' we
  will always mean the version of his model that is internal to
  extensional CIC as described above. The need for impredicativity is
  inherited from Reynolds' original construction, and is not a new
  requirement.}, internal to CIC.

After Reynolds' original paper, more abstract treatments of his ideas
were given by, \emph{e.g.}, Robinson and Rosolini \cite{rr94}, O'Hearn
and Tennent \cite{oht95}, Dunphy and Reddy \cite{dr04}, and Ghani {\em
  et al.} \cite{param_johann}. The approach is to use a categorical
structure --- reflexive graph categories for \cite{rr94,oht95,dr04}
and fibrations for \cite{param_johann} --- to represent sets and
relations, and to interpret types as appropriate functors and terms as
natural transformations. In particular,~\cite{dr04} aims to
``[address] parametricity in all its incarnations'', and similarly for
\cite{param_johann}. Surprisingly and significantly, however,
Reynolds' original model does \emph{not} arise as a direct instance of
either framework. This leads us to ask:

\vspace*{0.005in}
\begin{center}
{\em What constitutes a \emph{good} framework for relational parametricity?}
\end{center}
\vspace*{0.005in}

\noindent Our answer is that such a framework should:

\vspace*{0.005in}
\begin{center}
\begin{enumerate}
\em\item Deliver a relationally parametric model for each
instantiation of its parameters, from which it uniformly produces such
models. In particular, it should allow a choice of a suitable
meta-theory (the Calculus of Inductive Constructions, the theory of
$\omega$-sets, etc.).
\item Admit the two canonical relationally parametric models mentioned
  above as direct instances in a natural, uniform way.
\item Abstractly formulate properties that good models of parametricity for System F
  should be expected to satisfy.
\end{enumerate}
\end{center}

Criterion 1 ensures that we indeed get a true {\em framework} rather
than just a reusable blueprint for constructing models of
parametricity. Criterion 2 remains unsatisfied for the frameworks of
Dunphy and Reddy and of Ghani {\em et al.} because Reynolds' original
model formulated syntactically does not satisfy certain strictness
conditions imposed by \cite{dr04,param_johann}. For example, let
$\alpha \vdash S(\alpha)$ and $\alpha \vdash T(\alpha)$ be two types,
with object interpretations $S_0$ and $T_0$ and relational
interpretations $S_1$ and $T_1$. The interpretation of the product
$\alpha \vdash S(\alpha) \times T(\alpha)$ should be an appropriate
product of interpretations; that is, the object interpretation should
map a set $A$ to $S_0(A) \times T_0(A)$ and the relational
interpretation should map a relation $R$ to $S_1(R) \times T_1(R)$,
with the product of two relations defined in the obvious way. For the
Identity Extension Lemma to hold, we need $S_1(\Eq(A)) \times
T_1(\Eq(A))$ to be the same as $\Eq(S_0(A) \times T_0(A))$. Here, the
equality relation $\Eq(A)$ on a set $A$ maps $(a,b) : A \times A$ to
the type $\Id(a,b)$ of proofs of equality between $a$ and $b$, so that
$a$ and $b$ are related iff $\Id(a,b)$ is inhabited, {\em i.e.}, iff
$a$ is identical to $b$. By the induction hypothesis, $S_1(\Eq(A))$ is
$\Eq(S_0(A))$, and similarly for $T$, so we need to show that
$\Eq(S_0(A)) \times \Eq(T_0(A))$ is $\Eq(S_0(A) \times T_0(A))$. But
this is not necessarily the case since the identity type on a product
is in general not \emph{identical} to the product of identity types,
but rather just suitably \emph{isomorphic}. So the interpretation of
$\alpha \vdash S(\alpha) \times T(\alpha)$ is not necessarily an
indexed or fibered functor (in the settings of \cite{dr04} and
\cite{param_johann}, respectively).

Three ways to fix this problem come to mind. Firstly, we can attempt
to change the meta-theory, by, \emph{e.g.}, imposing an additional
axiom asserting that two logically equivalent propositions are
definitionally equal. We do not pursue this approach here: the goal of
our framework is to \emph{directly subsume the important models in
  their natural meta-theories}, as per criteria 1 and 2 above, rather
than require the user to augment the meta-theory with {\em ad hoc}
axioms to make the shoe fit. The second possibility is to use the
syntactic analogue of strictification, pursued in, \emph{e.g.},
\cite{AGJ14}. The idea is that instead of interpreting a closed type
as a set $A$ (on the object level), we interpret it as a set $A$
endowed with a relation $R_A$ that is \emph{isomorphic}, but not
necessarily identical, to the canonical discrete relation $\Eq_A$. The
chosen equality relation on the set $A$ --- more precisely, on the
entire structure $(A,(R_A,i : R_A \simeq \Eq_A))$ --- will then be
$R_A$ rather than $\Eq_A$. This allows us to construct $R_A$ in a way
that respects all type constructors on the nose, so that the
aforementioned issue with $\Eq(S_0(A)) \times \Eq(T_0(A))$ not being
identical to $\Eq(S_0(A) \times T_0(A))$ is avoided. The problem,
however, is that there can be many different ways to endow $A$ with a
discrete relation $(R_A,i)$; in other words, the type of discrete
relations on $A$ is not contractible. It is thus unclear whether and
how this ``discretized'' version of Reynolds' model is equivalent to
the original, intended one.

Here we suggest a third approach: we record the isomorphisms
witnessing the preservation of the Identity Extension Lemma for each
type constructor, and propagate them through the construction. This
means, however, that we can no longer interpret a type $\alpha \vdash
T(\alpha)$ as a pair of maps $T_0 : |\Set| \to \Set$ and $T_1 : |\rel|
\to \rel$; indeed, since the domain of $T_1$ is the discrete category
$|\rel|$, $T_1$ is not required to preserve isomorphisms in
$\rel$. As a result, even if we know that the pair $(T_0,T_1)$
satisfies the Identity Extension Lemma, its reindexing --- defined by
precomposition --- might not. The upshot is that the obvious
``$\lambda$2-fibration'' corresponding to Reynolds' original model is
not necessarily a fibration at all.

We solve this problem by specifying subcategories $\M(0) \subseteq
\Set$ and $\M(1) \subseteq \rel$ of \emph{relevant isomorphisms} that
form a \emph{reflexive graph category with isomorphisms}. Abstractly,
this structure gives us two \emph{face maps} (called $\partial_0$ and
$\partial_1$ in \cite{dr04}), which represent the domain and codomain
projections, and a \emph{degeneracy} (called $I$ in \cite{dr04}),
which represents the equality functor. We interpret a type $\alpha
\vdash T(\alpha)$ as a pair of functors $T_0 : \M(0) \to \M(0)$ and
$T_1 : \M(1) \to \M(1)$ that together comprise a \emph{face map- and
  degeneracy-preserving reflexive graph functor}, and interpret each
term as a \emph{face map- and degeneracy-preserving reflexive graph
  natural transformation}.

Since the domain of $T_1$ is $\M(1)$, $T_1$ preserves all relevant
isomorphisms between relations, so the reindexing of $(T_0,T_1)$ is
now well-defined. Choosing $\M(1)$ to contain the isomorphism between
the two relations $\Eq(S_0(A)) \times \Eq(T_0(A))$ and $\Eq(S_0(A)
\times T_0(A))$ yields the satisfaction of the Identity Extension
Lemma for products; other type constructors follow the same
pattern. We note that although the preservation of isomorphisms on the
\emph{relation} level is sufficient to carry out the model
construction, we formally require the preservation of relevant
isomorphisms on the \emph{object} level, too. This makes the framework
more uniform and, moreover, leads to the novel notion of a
\emph{categorical Reynolds' model}, in which interpretations of types
are endowed with a functorial action on isomorphisms and all
polymorphic functions respect this action. Furthermore, we go one
level higher and use the ideas of Orsanigo \cite{orsanigo_thesis} (and
Ghani {\em et al.}~\cite{param_ghani}, which it supersedes) to define
a \emph{proof-relevant categorical Reynolds' model}, in which,
additionally, witnesses of relatedness are themselves suitably related
via a yet higher relation.

This ``2-parametric'' model of course does not arise as an instance of
our framework since it requires additional structure --- \emph{e.g.},
the concept of a \emph{2-relation} --- pertaining to the higher notion
of parametricity. Nevertheless, we would still like to be able to
recognize it as a model parametric in the ordinary sense. Various
definitions of parametricity for models of System F exist:
\cite{dr04,param_johann} are examples of ``internal'' approaches to
parametricity, where a model is considered parametric if it is
produced via a specified procedure that bakes in desired features of
parametricity such as the Identity Extension Lemma. On the other hand,
\cite{gfs16,mr92,rr94,jac99} are examples of ``external'' approaches
to parametricity, in which reflexive graphs of models are used to
endow models of interest with enough additional structure that they
can reasonably be considered parametric. Surprisingly though, the
proof-relevant model we give in Section~\ref{sec:proof-relevant} does
not appear to satisfy any of these definitions, and in particular does
not satisfy any of the external ones. The ability to construct a
classifying reflexive graph seems to rely on an implicit assumption of
proof-irrelevance, which we elaborate on in
Section~\ref{sec:proof-relevant}. However, we propose a new definition
of a \emph{relationally parametric model of System F} in
Section~\ref{sec:main} and show that it subsumes not only the two
canonical parametric models of System F, but also the two novel ones
we give in this paper. In particular, it subsumes the proof-relevant
model given in Section~\ref{sec:proof-relevant}. \smallskip

\noindent The main contributions of this paper are as follows:
\emph{\begin{itemize}
\item We demonstrate that existing frameworks for the functorial
  semantics of relational parametricity for System F fail to directly
  subsume both canonical models of relational parametricity for System
  F.
\item We solve this problem by developing a {\em{good}} abstract
  framework for relational parametricity that allows a choice of
  meta-theory, delivers both canonical relationally parametric models
  of System F as direct instances in a uniform way, and exposes
  properties that good models of System F parametricity should be
  expected to satisfy, \emph{e.g.}, guaranteeing that interpretations
  of terms, not just types, suitably commute with the degeneracy.
 \item We give a novel definition of a \emph{parametric model of
   System F}, which is a hybrid of the external and internal
   approaches, and show that it subsumes both canonical models
   (expressed as instances of our framework).
 \item We give two novel relationally parametric models of System F
   --- one of which is proof-relevant and can be seen as parametric in
   a higher sense (``2-parametric'') --- and show that our definition
   recognizes both of these in a natural way, with the
   proof-irrelevant model arising as a direct instance of our
   framework.
\end{itemize}}

\paragraph{Fibrational Preliminaries}\label{sec:toolbox}
We give a brief introduction to fibrations, mainly to settle
notation. More details can be found in, {\em e.g.},~\cite{jac99}.

\begin{definition}
Let $U:\mathcal{E}\rightarrow \mathcal{B}$ be a functor.  A morphism
$g:Q\rightarrow P$ in $\mathcal{E}$ is \emph{cartesian} over
$f:X\rightarrow Y$ in $\mathcal{B}$ if $Ug=f$ and, for every
$g':Q'\rightarrow P$ in $\mathcal{E}$ with $Ug' = f \circ v$ for some
$v:UQ'\rightarrow X$, there is a unique $h:Q'\rightarrow Q$ with
$Uh=v$ and $g' = g \circ h$.  A functor $U:\mathcal{E}\rightarrow
\mathcal{B}$ is a \emph{fibration} if, for every object $P$ of
$\mathcal{E}$ and morphism $f:X\rightarrow UP$ of $\mathcal{B}$, there
is a cartesian morphism in $\mathcal{E}$ with codomain $P$ over $f$.
\end{definition}

If $U:\mathcal{E}\rightarrow \mathcal{B}$ is a fibration then
$\mathcal{E}$ is its \emph{total category} and $\mathcal{B}$ is its
\emph{base category}. An object $P$ in $\mathcal{E}$ is \emph{over}
its image $UP$, and similarly for morphisms. A morphism is {\em
  vertical} if it is over an identity morphism.  We write
$\mathcal{E}_X$ for the \emph{fiber over} an object $X$ in
$\mathcal{B}$, {\em i.e.}, the subcategory of $\mathcal{E}$ of objects
over $X$ and morphisms over $\id_X$.

If $U : \E \to \B$ is a fibration, we call a cartesian morphism over
$f$ with codomain $P$ a {\em cartesian lifting} of $f$ with codomain
$P$ with respect to $U$. A cartesian lifting of $f$ with codomain $P$
with respect to $U$ need not be unique, but it is always unique up to
vertical isomorphism. We are interested in fibrations in which
representative cartesian liftings are specified, or chosen.

\begin{definition}\label{def:specified}
A fibration $U : \E \to \B$ {\em is cloven} if it comes with a {\em
  choice} of cartesian liftings, {\em i.e.}, with one cartesian
lifting of $f$ with codomain $P$ with respect to $U$ regarded as
primary amongst all such cartesian liftings for each morphism $f$ in
$\B$ and object $P$ in $\E$.
\end{definition}

We emphasize that the choice of cartesian liftings is part of the
structure that is given when a fibration is cloven. In this case one uses the phrase ``the cartesian lifting'' of $f$ with codomain $P$ to refer to the chosen such lifting, which we denote by
$\cart{f}_P$. Any time we consider categorical objects ({\em e.g.},
categories, functors, etc.)  with particular structures ({\em e.g.},
products, adjoints, etc.) in this paper, we intend that those
structures are chosen in this sense.

The function mapping each object $P$ of $\mathcal{E}$ to the domain
$f^*P$ of $\cart{f}_P$ then extends to a functor $f^*: \mathcal{E}_{Y}
\rightarrow \mathcal{E}_{X}$ mapping each morphism $k:P\rightarrow P'$
in $\mathcal{E}_{Y}$ to the unique morphism $f^*k$ such that $k \circ
\cart{f}_{P} = \cart{f}_{P'} \circ f^*k$.  The universal property of
$\cart{f}_{P'}$ ensures the existence and uniqueness of $f^*k$.  We
call $f^*$ the \emph{substitution functor along $f$}.  We will be
especially interested in cloven fibrations whose substitution functors
are well-behaved:

\begin{definition}
A cloven fibration $U : \E \to \B$ is {\em split} if its substitution
functors are such that $\id^*= \id$ and $(g \circ f)^* = f^* \circ
g^*$.
\end{definition}

\noindent We will later require even more structure of our split fibrations:

\begin{definition}
A split fibration $U : \E \to \B$ has a {\em split generic object} if
there is an object $\Omega$ in $\B$, together with a collection of
isomorphisms $\theta_X$ mapping each morphism from $X$ to $\Omega$ in
$\B$ to an object of the fiber $\E_X$ that is natural in $X$, {\em
  i.e.}, is such that $\theta_Y(fg) = g^*(\theta_X(f))$ for every $f :
X \to \Omega$ and $g : Y \to X$.
\end{definition}

Seely~\cite{see87} gave a sound categorical semantics of System F in
$\lambda 2$-fibrations (presented as PL-categories). We will make good
use of this result below.


\section{Reflexive Graph Categories}\label{sec:rg-categories}
Although Reynolds himself showed that his original approach to
relational parametricity does not work in set theory, we can still use
it as a guide for designing an abstract framework for
parametricity. Instead of sets and relations, we consider abstract
notions of ``sets'' and ``relations'', and require them to be
related as follows: \emph{i)} for any relation $R$, there are two
canonical ways of projecting an object out of $R$, corresponding to
the domain and codomain operations, \emph{ii)} for any object $A$,
there is a canonical way of turning it into a relation, corresponding
to the equality relation on $A$, and \emph{iii)} if we start with an
object $A$, turn it into a relation according to \emph{ii)}, and then
project out an object according to \emph{i)}, we get $A$ back. This
suggests that our abstract relations and the canonical operations on
them can be organized into a reflexive graph structure: categories
$\X_0$, $\X_1$ and functors $\fm_\top, \fm_\bot : \X_1 \to \X_0$, $\dm
: \X_0 \to \X_1$ such that $\fm_\top \circ \dm = \id = \fm_\bot \circ
\dm$, as is done in \cite{dr04}.

Since there are no set-theoretic models of System F
(\cite{not-set-theoretic}), all of the reflexive graph structure
identified above must to be internal to some ambient category
$\Ccal$. In particular, $\X_0$ and $\X_1$ must be categories internal
in $\Ccal$, and $\fm_\top$, $\fm_\bot$, and $\dm$ must be functors
internal in $\Ccal$. For Reynolds' original model, the ambient
category has types $A : \Ucal_1$ as objects and terms $f : \Sigma_{A,B
  : \Ucal_1}\,A \to B$ as morphisms. Here, $\Ucal_1$ is the universe
one level above the impredicative universe $\Ucal_0$; we will denote $\Ucal_0$ simply by $\Ucal$ below. This ensures that $\Ucal$ is
an object in $\Ccal$. To model relations, we introduce:
\begin{align*}
& \mathsf{isProp}(A) \coloneqq \Pi_{a,b:A} \, \Id(a,b) \\
& \mathsf{Prop} \coloneqq \Sigma_{A : \Ucal}\,\mathsf{isProp}(A)
\end{align*}
The type $\Prop$ of \emph{propositions} singles out those types in
$\Ucal$ with the property that any two inhabitants, if they exist, are
equal. Propositions can be used to model relations as follows: in
Reynolds' original model, $a : A$ is related to $b : B$ in at most one
way under any relation $R$ (either $(a,b) \in R$ or not), so the type
of proofs that $(a,b) \in R$ is a proposition. Conversely, given $R :
A \times B \to \mathsf{Prop}$, we consider $a$ and $b$ to be related
by $R$ iff $R(a,b)$ is inhabited.

To see the universe $\Ucal$ as a category $\Set$ internal to $\Ccal$
we take its object of objects $\Set_0$ to be $\Ucal$ and define its
object of morphisms by $\Set_1 \coloneqq \Sigma_{A,B : \Ucal}\,A \to
B$. We define the category $\mathsf{R}$ of relations by giving its
objects $\mathsf{R}_0$ and $\mathsf{R}_1$ of objects and morphisms,
respectively:
\begin{align*}
& \mathsf{R}_0 \coloneqq \Sigma_{A,B:\Set}~A \times B \to \Prop \\
& \mathsf{R}_1 \coloneqq \Sigma_{((A_1,A_2),R_A),((B_1,B_2),R_B) :
    \mathsf{R}_0} \Sigma_{(f,g) :  (A_1 \to B_1) \times (A_2 \to B_2)}
  \\  
& \;\;\;\;\;\;\;\;\; \Pi_{(a_1,a_2):A_1 \times A_2} R_A(a_1,a_2) \to
  R_B(f(a_1),g(a_2))   
\end{align*}

We clearly have two internal functors from $\mathsf{R}$ to $\Set$
corresponding to the domain and codomain projections, respectively. We also have an
internal functor $\Eq$ from $\Set$ to $\mathsf{R}$ that constructs an
equality relation with $\Eq \,A \coloneqq ((A,A),\Id_A)$ and $\Eq \,
((A,B),f) \coloneqq \big((\Eq \,A,\Eq\,B),(f,f),\mathsf{ap}_f\big)$.
Here, the term $\mathsf{ap}_f : \Id_A (a_1,a_2) \to \Id_B(f(a_1),
f(a_2))$ is defined as usual by $\Id$-induction and witnesses the
fact that $f$ respects equality.

These observations motivate the next two definitions, in which we
denote the category of categories and functors internal to $\Ccal$ by
$\mathsf{Cat}(\Ccal)$, and assume $\Ccal$ is locally small and has all
finite products. (A category is {\em locally small} if each of its
hom-sets is small, {\em i.e.}, is a set rather than a proper
class.) 
\begin{definition}\label{def:rg}
A \emph{reflexive graph structure} $\X$ on a category $\Ccal$ consists of:
\begin{itemize}
\item objects $\X(0)$ and $\X(1)$ of $\Ccal$
\item distinct arrows $\X(\fm_\star) : \X(1) \to \X(0)$ for $\star : \two$
\item an arrow $\X(\dm) : \X(0) \to \X(1)$
\end{itemize}
such that $\X(\fm_\star) \circ \X(\dm) = \id$.
\end{definition}
The requirement that the two face maps $\X(\fm_\top)$ and
$\X(\fm_\bot)$ are distinct is to ensure that there are enough
relations for the notion of relation-preservation to be
meaningful. Otherwise, as also observed in \cite{dr04}, we could see
\emph{any} category $C$ as supporting a trivial reflexive graph
structure whose only relations are the equality ones. For readers
familiar with \cite{jac99}, the condition $\X(\fm_\top) \neq
\X(\fm_\bot)$ serves a purpose similar to that of the requirement in
Definition 8.6.2 of \cite{jac99} that the fiber category
$\mathbb{F}_1$ over the terminal object in $\mathbb{C}$ is the
category of relations in the preorder fibration $\mathbb{D} \to
\mathbb{E}$ on the fiber category $\mathbb{E}_1$ over the terminal
object in $\mathbb{B}$. Both conditions imply that some relations must
be \emph{heterogeneous}. But while in \cite{jac99} relations are
obtained in a standard way as predicates (given by a preorder
fibration) over a product, we do not assume that relations are
constructed in any specific way, but rather only that the abstract
operations on relations suitably interact. Moreover, since the two
face maps $\X(\fm_\top)$ and $\X(\fm_\bot)$ are distinct, any morphism
generated by the face maps and the degeneracy $\X(\dm)$ must be one of
the seven distinct maps $\id_{\X(0)},\id_{\X(1)}, \X(\fm_\star),
\X(\dm)$, and $\X(\dm) \circ \X(\fm_\star)$ for $\star : \two$. Every
such morphism thus has a canonical representation.
 
\begin{definition}\label{def:rg-category}
A \emph{reflexive graph category (on $\Ccal$)} is a reflexive graph
structure on $\mathsf{Cat}(\Ccal)$.
\end{definition}
\begin{example}[PER model]\label{ex:per}
We take the ambient category $\Ccal$ to be the category of
$\omega$-sets, given in~Definition~6.3 of~\cite{longo_moggi}. We
construct a reflexive graph category, which we call
$\mathcal{R}_\mathit{PER}$, as follows. The internal category
$\mathcal{R}_\mathit{PER}(0)$ of ``sets'' is the category
$\mathbf{M}'$ given in Definition 8.4
of~\cite{longo_moggi}. Informally, the objects of $\mathbf{M}'$ are
partial equivalence relations on $\nat$, and the morphisms are
realizable functions that respect such relations. To define the
internal category $\mathcal{R}_\mathit{PER}(1)$ of ``relations'', we
first construct its object of objects. The carrier of this
$\omega$-set is the set of pairs of the form $R \coloneqq
((A_\mathsf{d}, A_\mathsf{c}), R_A)$, where $A_\mathsf{d}$ and
$A_\mathsf{c}$ are partial equivalence relations and $R_A$ is a
saturated predicate on the product PER $A_\mathsf{d} \times
A_\mathsf{c}$. A saturated predicate on a PER $A$ is a predicate on
$\nat$ such that $a_1 \sim_A a_2$ and $R(a_1)$ imply $R(a_2)$. To
finish the construction of our object of objects for
$\mathcal{R}_\mathit{PER}(1)$ we take any pair $((A_\mathsf{d},
A_\mathsf{c}), R_A)$ as above to be realized by any natural number.

The carrier of the object of morphisms for
$\mathcal{R}_\mathit{PER}(1)$ comprises all pairs of the form
\[\big(\big(((A_\mathsf{d},A_\mathsf{c}),R_A), ((B_\mathsf{d},B_\mathsf{c}),R_B)\big),
\big(\{m_1\}_{A_\mathsf{d} \to B_\mathsf{d}}, \{m_2\}_{A_\mathsf{c} \to B_\mathsf{c}}\big)\big)\]
satisfying the condition that, for any $k$, $l$ such that $k \sim_{A_\mathsf{d}}\! k$, $l \sim_{A_\mathsf{c}}\! l$, and $R_A(\langle k,l \rangle)$ holds, $R_B\big(\langle m_1 \cdot k, m_2
\cdot l \rangle\big)$ holds as well. The first component records the
domain and codomain of the morphism and the second component is a pair
of equivalence classes under the specified exponential PERs. As in
\cite{longo_moggi}, we denote the application of the $n^{\mathit{th}}$
partial recursive function to a natural number $a$ in its domain by $n
\cdot a$. To finish the construction of the object of morphisms for
$\mathcal{R}_\mathit{PER}(1)$, we take a pair of pairs as above to be
realized by a natural number $k$ iff $\fst(k) \sim_{A_\mathsf{d} \to B_\mathsf{d}} m_1$
and $\snd(k) \sim_{A_\mathsf{c} \to B_\mathsf{c}} m_2$.

We again have two internal functors
$\mathcal{R}_\mathit{PER}(\fm_\top)$ and
$\mathcal{R}_\mathit{PER}(\fm_\bot)$ from
$\mathcal{R}_\mathit{PER}(1)$ to $\mathcal{R}_\mathit{PER}(0)$ corresponding to the two projections. We
also have an equality functor $\mathsf{Eq}$ from
$\mathcal{R}_\mathit{PER}(0)$ to $\mathcal{R}_\mathit{PER}(1)$ whose
action on objects is given by $\mathsf{Eq}\,A \coloneqq ((A,A),\Delta_A)$, where $\Delta_A(k)$ iff $\fst(k) \sim_A \snd(k)$,
and whose action on morphisms is given by
\[ \mathsf{Eq}\,((A,B),\{m\}_{A \to B}) \coloneqq \big((\Eq\,A,\Eq\,B), (\{m\}_{A \to B},\{m\}_{A \to B})\big) \] 
\end{example}
\begin{example}[Reynolds' model]\label{ex:reynolds}
We obtain a reflexive graph category $\mathcal{R}_\mathit{REY}$ by
taking $\mathcal{R}_\mathit{REY}(0) := \set$,
$\mathcal{R}_\mathit{REY}(1) := \mathsf{R}$, and
$\mathcal{R}_\mathit{REY}(\dm) \coloneqq \Eq$, and letting
$\mathcal{R}_\mathit{REY}(\fm_\top)$ and
$\mathcal{R}_\mathit{REY}(\fm_\bot)$ be the functors corresponding to
the domain and codomain projections, respectively.
\end{example}
If $\X$ is a reflexive graph category, then the discrete graph
category $|\X|$ and the product reflexive graph category $\X^n$ for $n
\in \nat$ are defined in the obvious ways: $|\X(l)|$ has the same
objects as $\X(l)$ but only the identity morphisms, and $(\X \times
\Y)(l) = \X(l) \times \Y(l)$ for $l \in \{0,1\}$. For the latter, the
product on the right-hand side is a product of internal categories,
which exists because $\Ccal$ has finite products by assumption.

If $C$ is an internal category, we denote by $C_0$ and $C_1$ the objects of $\C$ representing the objects and morphisms of $C$, respectively. If $F : C \to D$ is an internal functor, we denote by $F_0 : C_0 \to D_0$ and $F_1 : C_1 \to D_1$ the arrows of $\C$ representing the object and morphisms parts of $F$, respectively. Also:

\begin{notation}
We will use the following notation with respect to an internal category $C$ in $\C$:
\begin{itemize}
\item Given a ``generalized object" $a : J \to C_0$ (with $J$ arbitrary), we denote by $\id_C[a]$ the arrow $\mathsf{id}_C \circ a$, where $\mathsf{id}_C : C_0 \to C_1$ is the arrow representing identity morphisms in $C$. 

\item For a ``generalized morphism"  $f : J \to C_1$ (with $J$ arbitrary), we denote by $\source_C[f]$ and $\target_C[f]$ the arrows $\source_C \circ f$ and $\target_C \circ f$ respectively, where $\source_C, \target_C : C_1 \to C_0$ are the arrows representing the source and target operations in $C$.

\item For generalized morphisms $f,g : J \to C_1$ such that $\target_C[f] = \source_C[g]$, we denote by $g \circ_C f$ the arrow $\mathsf{comp}_C \circ \langle f, g\rangle$, where $\mathsf{comp}_C : \textit{pullback}(\target_C,\source_C)\to C_1$ is the arrow representing composition in $C$, its domain $\textit{pullback}(\target_C,\source_C)$ is the pullback of the two arrows $\target_C$ and $\source_C$, and $\langle f, g\rangle$ is the canonical morphism into this pullback.

\item We say that $f : J \to C_1$ is an isomorphism if there exists a $g : J \to C_1$ such that $\source_C[f] = \target_C[g]$, $\source_C[g] = \target_C[f]$ and $f \circ_C g = \id_C[\source_C[g]]$, $g \circ_C f = \id_C[\source_C[f]]$. If such a $g$ exists, it is necessarily unique and hence will be denoted by $f^{-1}$.   
\end{itemize}
\end{notation}

Given a reflexive graph category $\X$ axiomatizing the sets and
relations, an obvious first attempt at pushing Reynolds' original idea
through is to take the interpretation $\sem{T}$ of a type
$\overline{\alpha} \vdash T$ with $n$ free type variables to be a pair
$(\sem{T}(0), \sem{T}(1))$, where $\sem{T}(0) : |\X(0)|^n \to \X(0)$
and $\sem{T}(1) : |\X(1)|^n \to \X(1)$ are functions giving the
``set'' and ``relation'' interpretations of the type $T$. Although as
explained in the introduction, this approach will need some tweaking
--- we will need to endow $\sem{T}(0)$ and $\sem{T}(1)$ with actions
on \emph{some} morphisms --- it suggests:

\begin{definition}\label{def:rg-functor}
Let $\X$ and $\Y$ be reflexive graph categories. A {\em reflexive
  graph functor} $\F : \X \to \Y$ is a pair $(\F(0),\F(1))$ of
functors such that $\F(0): \X(0) \to \Y(0)$ and $\F(1):\X(1) \to
\Y(1)$.
\end{definition}
Writing $T_0$ for $\sem{T}(0)$ and $T_1$ for $\sem{T}(1)$, we recall
from the introduction that $T_0$ and $T_1$ should be appropriately
related via the domain and codomain projections and the equality
functor. Since the two face maps $\X(\fm_\star)$ now model the
projections, and the degeneracy $\X(\dm)$ models the equality functor,
we end up with the following conditions:
\begin{itemize}
\item[\emph{i)}] for each object $\overline{R}$ in $\X(1)^n$, we have $\X(\fm_\star) \,
T_1(\overline{R}) = T_0(\X(\fm_\star)^n \, \overline{R})$
\item[\emph{ii)}] for each object $\overline{A}$ in $\X(0)^n$, we have
$\X(\dm) \, T_0(\overline{A}) = T_1(\X(\dm)^n \, \overline{A})$
\end{itemize}

We examine what these conditions imply for Reynolds' model by considering
the product $\alpha \vdash S(\alpha) \times T(\alpha)$ of two types
$\alpha \vdash S(\alpha)$ and $\alpha \vdash T(\alpha)$.  By the
induction hypothesis, $S$ and $T$ are interpreted as pairs $(S_0,S_1)$
and $(T_0,T_1)$, where $S_0,T_0 : \Set_0 \to \Set_0$ and $S_1,T_1 :
\mathsf{R}_0 \to \mathsf{R}_0$ satisfy \emph{i)} and \emph{ii)}. The
interpretation of a product should be a product of interpretations,
{\em i.e.}, $(S \times T)_0\,A \coloneqq S_0(A) \times T_0(A)$ and $(S
\times T)_1 \,R \coloneqq S_1(R) \times T_1(R)$. It remains to be seen that this interpretation satisfies \emph{i)} and
\emph{ii)}. Fix a relation $R$ on $A$ and $B$. Condition \emph{i)} entails that
$S_1(R) \coloneqq ((S_0(A),S_0(B)),R_S)$ and $T_1(R) \coloneqq
((T_0(A),T_0(B)),R_T)$ for some $R_S$ and $R_T$. Thus $S_1(R) \times
T_1(R)$ has the form $\big((S_0(A) \times T_0(A), S_0(B) \times
T_0(B)), R_{S \times T}\big)$, where $R_{S \times T}$ maps a pair of
pairs $((a,b),(c,d))$ to $R_S(a,c) \times R_T(b,d)$. Thus
\emph{i)} is satisfied simply by construction, which leads us to define:\begin{definition}\label{def:rgfmp}
 A reflexive graph functor $\F : \X \to \Y$ is {\em face map-preserving} if the following diagram in $\cat(\C)$ commutes for all $\star \in \two$:
\begin{center}
\scalebox{0.9}{
\begin{tikzpicture}
\node (N0) at (0,1.5) {$\X(1)$};
\node (N1) at (0,0) {$\X(0)$};
\node (N2) at (3,1.5) {$\Y(1)$};
\node (N3) at (3,0) {$\Y(0)$};
\draw[->] (N0) -- node[left]{$\X(\fm_\star)$} (N1);
\draw[->] (N0) -- node[above]{$\F(1)$} (N2);
\draw[->] (N1) -- node[below]{$\F(0)$} (N3);
\draw[->] (N2) -- node[right]{$\Y(\fm_\star)$} (N3);
\end{tikzpicture}}
\end{center}
\end{definition}

In Reynolds' model, condition \emph{ii)} gives that $S_1(\Eq(A))$ is
$\Eq(S_0(A))$ for any set $A$, and similarly for $T$. We thus need to
show that $\Eq(S_0(A)) \times \Eq(T_0(A))$ is $\Eq(S_0(A) \times
T_0(A))$. But while the domains and codomains of these two relations
agree (all are $S_0(A) \times T_0(A)$), the former maps
$((a,b),(c,d))$ to $\Id(a,c) \times \Id(b,d)$, while the latter maps
it to $\Id((a,b),(c,d))$. These two types are not necessarily identical, but they
are \emph{isomorphic} (\emph{i.e.}, there are functions back and forth
that compose to identity on both sides).

We thus relax condition \emph{ii)} to allow an isomorphism
$\varepsilon_T(\overline{A}) : \X(\dm) \, T_0(\overline{A}) \cong
T_1(\X(\dm)^n \, \overline{A})$. In fact, we can require more: since
the domains and codomains of $\X(\dm) \, T_0(\overline{A})$ and
$T_1(\X(\dm)^n \, \overline{A})$ coincide by condition \emph{i)}, we
can insist that both projections map the isomorphism
$\varepsilon_T(\overline{A})$ to the identity morphism on
$T_0(\overline{A})$. This coherence condition is a natural counterpart
to the equation $\X(\fm_\star) \circ \X(\dm) = \mathsf{id}$, and turns
out to be not just a design choice but a necessary requirement: in
Reynolds' model, for instance, the proof that the interpretations of
$\forall$-types (as defined later) suitably commute with the functor
$\Eq$ depends precisely on the morphisms underlying the maps
$\varepsilon_T(\overline{A})$ being identities. This suggests:

\begin{definition}\label{def:rgdp}
A reflexive graph functor $\F : \X \to \Y$ is {\em
  degeneracy-preserving} if the following diagram in $\cat(\C)$ commutes up to a given natural isomorphism $\varepsilon_\F$ satisfying the coherence condition $\Y(\fm_\star)_1 \circ \varepsilon_\F
= \id_{\Y(0)}[\F(0)_0]$ for $\star \in \two$:
\begin{center}
\scalebox{0.9}{
\begin{tikzpicture}
\node (N0) at (0,1.5) {$\X(0)$};
\node (N1) at (0,0) {$\X(1)$};
\node (N2) at (3,1.5) {$\Y(0)$};
\node (N3) at (3,0) {$\Y(1)$};
\draw[->] (N0) -- node[left]{$\X(\dm)$} (N1);
\draw[->] (N0) -- node[above]{$\F(0)$} (N2);
\draw[->] (N1) -- node[below]{$\F(1)$} (N3);
\draw[->] (N2) -- node[right]{$\Y(\dm)$} (N3);
\end{tikzpicture}}
\end{center}
\end{definition}

As a first approximation, we can try to interpret a type
$\overline{\alpha} \vdash T$ with $n$ free type variables as a face
map- and degeneracy-preserving reflexive graph functor $(T_0,T_1) :
|\X|^n \to \X$. Reynolds' original idea for interpreting terms
suggests that the interpretation of a term $\overline{\alpha}; x : S
\vdash t : T$ should be a (vacuously) natural transformation $t_0 :
S_0 \to T_0$. As observed in \cite{param_johann}, the Abstraction
Theorem can then be formulated as follows: there is a (vacuously)
natural transformation $t_1 : S_1 \to T_1$ such that, for any object
$\overline{R}$ in $\X(1)^n$, we have $\X(\fm_\star) \,
t_1(\overline{R}) = t_0(\X(\fm_\star)^n \, \overline{R})$. To see that
this does indeed give what we want, we revisit Reynolds' model.
There, the face maps are the domain and codomain projections and an
object $\overline{R}$ in $\X(1)^n$ is an $n$-tuple of relations.
Denote $\X(\fm_\top)^n \ \overline{R}$ by $\overline{A}$ and
$\X(\fm_\bot)^n \ \overline{R}$ by $\overline{B}$. Then
$t_1(\overline{R})$ is a morphism of relations from
$S_1(\overline{R})$ to $T_1(\overline{R})$ and, since $S_1$ and $T_1$
are face map-preserving, $S_1(\overline{R}) \coloneqq
\big((S_0(\overline{A}), S_0(\overline{B})), R_S\big)$ and
$T_1(\overline{R}) \coloneqq \big((T_0(\overline{A}),
T_0(\overline{B})), R_T\big)$ for some $R_S$ and $R_T$.  By
definition, $t_1(\overline{R})$ gives maps $f : S_0(\overline{A}) \to T_0(\overline{A})$, $g : S_0(\overline{B}) \to T_0(\overline{B})\big)$, together with a map $h : \Pi_{(a_1,a_2):
  S_0(\overline{A}) \times S_0(\overline{B})} R_S(a_1,a_2) \to
R_T(f(a_1),g(a_2))$ stating precisely that $f$ and $g$ map related
inputs to related outputs. By definition, $\X(\fm_\top) \, t_1(\overline{R})$
is $\big((S_0(\overline{A}), T_0(\overline{A})), f\big)$ and $\X(\fm_\bot) \,
t_1(\overline{R})$ is $\big((S_0(\overline{B}), T_0(\overline{B}), g\big)$, so the
condition that $\X(\fm_\star) \, t_1(\overline{R})$ is
$t_0(\X(\fm_\star)^n \,\overline{R})$ implies that the maps underlying
$t_0(\overline{A})$ and $t_0(\overline{B})$ must be $f$ and $g$,
respectively, and so must indeed map related inputs to related
outputs, as witnessed by $h$. Pairing the natural transformations $t_0$ and $t_1$
motivates:

\begin{definition}\label{def:rg-nat-trans}
Let $\F, \Gcal : \X \to \Y$ be reflexive graph functors. A
\emph{reflexive graph natural transformation} $\eta : \F \to \Gcal$ is
a pair $(\eta(0),\eta(1))$ of natural transformations $\eta(0): \F(0)
\to \Gcal(0)$ and $\eta(1):\F(1)\to\Gcal(1)$.
\end{definition} 
\noindent The Abstraction Theorem then further suggests defining:
\begin{definition}
A reflexive graph natural transformation $\eta : \F \to \Gcal$ between two face map-preserving reflexive graph functors is {\em face map-preserving} if for any $\star \in \two$ we have
\[ \Y(\fm_\star)_1 \circ \eta(1) = \eta(0) \circ \X(\fm_\star)_0 \]
\end{definition}
\noindent
The interpretation of a term $\overline{\alpha}; x : S \vdash t : T$
should then be a face map-preserving natural transformation from
$(S_0,S_1)$ to $(T_0,T_1)$. We also have the dual notion:
\begin{definition}\label{def:deg-pres-nat-trans}
A reflexive graph natural transformation $\eta : \F \to \Gcal$ between two degeneracy-preserving reflexive graph functors $(\F,\varepsilon_\F)$ and $(\G,\varepsilon_\G)$ is {\em degeneracy-preserving} if for any $\star \in \two$, we have
\[ (\eta(1) \circ \X(\dm)_0) \circ_{\Y(1)} \varepsilon_\F = \varepsilon_\G \circ_{\Y(1)} (\Y(\dm)_1 \circ \eta(0)) \]
Intuitively, the above equation represents the commutativity of the following diagram in the internal category $\Y(1)$:
\begin{center}
\scalebox{0.9}{
\begin{tikzpicture}
\node (N0) at (0,2) {$\Y(\dm)_0 \circ \F(0)_0$};
\node (N1) at (0,0) {$\Y(\dm)_0 \circ \G(0)_0$};
\node (N2) at (5,2) {$\F(1)_0 \circ \X(\dm)_0$};
\node (N3) at (5,0) {$\G(1)_0 \circ \X(\dm)_0$};
\draw[->] (N0) -- node[left]{$\Y(\dm)_1 \circ \eta(0)$} (N1);
\draw[->] (N0) -- node[above]{$\varepsilon_\F$} (N2);
\draw[->] (N1) -- node[below]{$\varepsilon_\G$} (N3);
\draw[->] (N2) -- node[right]{$\eta(1) \circ \X(\dm)_0$} (N3);
\end{tikzpicture}}
\end{center}
\end{definition}
There is no explicit analogue of
Definition~\ref{def:deg-pres-nat-trans} in Reynolds' model for the
following reason: Reynolds' model (as well as the PER model) is
proof-irrelevant, in the precise sense that the functor $\langle
\X(\fm_\bot), \X(\fm_\top) \rangle$ is faithful, and this condition is
sufficient to guarantee that \emph{any} face map-preserving natural
transformation is automatically degeneracy-preserving as well. This
may or may not be the case in proof-relevant models (although in the
model from Section~\ref{sec:proof-relevant} it is), so we explicitly
restrict attention below only to those natural transformations that
are face map- \emph{and} degeneracy-preserving (as also done in
\cite{dr04}), and omit further mention of these properties.

We have the usual laws of identity and composition of reflexive graph functors and natural transformations:

\begin{definition}\label{def:id_func}
Given a reflexive graph category $\X$, the identity reflexive graph functor $\mathsf{1}_\X : \X \to \X$ is defined as follows:
\begin{itemize}
\item $\mathsf{1}_\X(l)$ is the identity functor on $\X(l)$ 
\item $\varepsilon_{\mathsf{1}_\X} \coloneqq \id_{\X(1)}[\X(\dm)_0]$
\end{itemize}
\end{definition}

\begin{definition}\label{def:comp_func}
Given two reflexive graph functors $\F : \X \to \Y$ and $\G : \Y \to \Z$, let $\G \circ \F : \X \to \Z$ be the reflexive graph functor defined as follows:
\begin{itemize}
\item $(\G \circ \F)(l) \coloneqq \G(l) \circ \F(l)$ 
\item $\varepsilon_{\G \circ \F} \coloneqq (\G(1)_1 \circ \varepsilon_\F) \circ_{\Z(1)} (\varepsilon_{\G} \circ \F(0)_0)$
\end{itemize}
\end{definition}

\begin{definition}\label{def:id_nat}
Given a reflexive graph functor $\F : \X \to \Y$, the identity reflexive graph natural transformation $\mathsf{1}_\F : \F \to \F$ is defined by $\mathsf{1}_\F(l) = \id_{\Y(l)}[\F(l)_0]$.
\end{definition}

\begin{definition}\label{def:comp_nat}
Given reflexive graph functors $\F,\G,\mathcal{H} : \X \to \Y$ and reflexive graph natural transformations $\eta_1 : \F \to \G$ and $\eta_2 : \G \to \mathcal{H}$, let $\eta_2 \circ \eta_1 : \F \to \mathcal{H}$ be the reflexive graph natural transformation defined by $(\eta_2 \circ \eta_1)(l) \coloneqq \eta_2(l) \circ_{\Y(l)} \eta_1(l)$.
\end{definition}

\begin{definition}\label{def:comp_func_nat}
Given reflexive graph functors $\F : \X \to \Y$ and $\G_1,\G_2 : \Y \to \Z$, and a reflexive graph natural transformation $\eta : \G_1 \to \G_2$, let $\eta \circ \F : \G_1 \circ \F \to G_2 \circ \F$ be the reflexive graph natural transformation defined by $(\eta \circ \F)(l) \coloneqq \eta(l) \circ \F(l)_0$.
\end{definition}

\begin{definition}\label{def:comp_nat_func}
Given reflexive graph functors $\F_1,\F_2 : \X \to \Y$ and $\G : \Y \to \Z$, and a reflexive graph natural transformation $\eta : \F_1 \to \F_2$, let $\G \circ \eta : \G \circ \F_1 \to \G \circ \F_2$ be the reflexive graph natural transformation defined by $(\G \circ \eta)(l) \coloneqq \G(l)_1 \circ \eta(l)$.
\end{definition}

One basic example of a reflexive graph functor which will be used often and will end up interpreting type variables is the projection:
\begin{definition}
Given a reflexive graph category $\X$ and $0 \leq i < n$, the ``$i$-th projection" reflexive graph functor $\mathsf{pr}^n_i : \X^n \to \X$ is defined as follows:
\begin{itemize}
\item $\mathsf{pr}^n_i(l)$ is the internal functor projecting out the $i$-th component
\item $\varepsilon_{\mathsf{pr}^n_i} \coloneqq \id_{\X(1)}\big[\X(\dm)_0 \circ \mathsf{pr}^n_i(0)_0\big]$
\end{itemize}
\end{definition}
\noindent Dually, we have the following:

\begin{definition}\label{def:tuple_func}
Given reflexive graph functors $\F_0, \ldots, \F_{m-1} : \X \to \Y$, let $\langle \F_0, \ldots,$ $\F_{m-1}\rangle$ be the reflexive graph functor from $\X$ to $\Y^m$ defined as follows:
\begin{itemize}
\item $\langle \F_0, \ldots, \F_{m-1}\rangle(l) \coloneqq \langle \F_0(l), \ldots, \F_{m-1}(l)\rangle$
\item $\varepsilon_{\langle \F_0, \ldots, \F_{m-1}\rangle} \coloneqq \langle \varepsilon_{\F_0}, \ldots, \varepsilon_{\F_{m-1}}\rangle$
\end{itemize}
\end{definition}

\begin{definition}\label{def:tuple_nat}
Given reflexive graph functors $\F_0, \ldots, \F_{m-1}, \G_0, \ldots, \G_{m-1} : \X \to \Y$, and reflexive graph natural transformations $\eta_0 : \F_0 \to \G_0, \ldots, \eta_{m-1} : \F_{m-1} \to \G_{m-1}$, let $\langle \eta_0, \ldots, \eta_{m-1}\rangle : \langle \F_0, \ldots, \F_{m-1}\rangle \to \langle \G_0, \ldots, \G_{m-1}\rangle$ be the reflexive graph natural transformation defined by $\langle \eta_0, \ldots, \eta_{m-1}\rangle(l) \coloneqq \langle \eta_0(l), \ldots, \eta_{m-1}(l)\rangle$.
\end{definition}

\begin{lemma}\label{lem:2catprop}
We have the following properties:
\begin{enumerate}
\item The identity reflexive graph functor serves as the identity for the composition of reflexive graph functors.
\item The composition of reflexive graph functors is associative.
\item The identity reflexive graph natural transformation serves as the identity for the composition of reflexive graph natural transformations.
\item The composition of reflexive graph natural transformations is associative.
\item The composition $(-) \circ \F$ of a reflexive graph functor and a reflexive graph natural transformation is functorial. 
\item The composition $\G \circ (-)$ of a reflexive graph natural transformation and a reflexive graph functor is functorial.
\end{enumerate}
\end{lemma}


\section{Reflexive Graph Categories with Isomorphisms}
As noted above, if we try to interpret a type $\overline{\alpha}
\vdash T$ as a reflexive graph functor $\sem{T} : \X^n \to \X$ we
encounter a problem with contravariance. Specifically, if $\alpha
\vdash A$ and $\alpha \vdash B$ are types, then to interpret the
function type $\alpha \vdash A \to B$ as the exponential of $\sem{A}$
and $\sem{B}$, $\sem{A \to B}(0)$ must map each object $X$ to the
exponential $(\sem{A}(0) \, X) \Rightarrow (\sem{B}(0) \, X)$ and each
morphism $f : X \to Y$ to a morphism from $(\sem{A}(0) \, X)
\Rightarrow (\sem{B}(0) \, X)$ to $(\sem{A}(0) \, Y) \Rightarrow
(\sem{B}(0) \, Y)$. But there is no canonical way to construct a morphism of this type because $\sem{A}(0) \,f$ goes in the wrong
direction. This is a well-known problem that is unrelated to parametricity. 

The usual solution is to require the domains of the functors
interpreting types to be discrete, so that $\sem{T} : |\X|^n \to
\X$. However, as noted in the introduction, this will not work in our
setting. Consider types $\alpha \vdash S(\alpha)$ and $\cdot \vdash
T$. By the induction hypothesis, $\sem{S} : |\X| \to \X$ and $\sem{T} : \mathsf{1} \to \X$ are face map- and degeneracy-preserving reflexive graph
functors. The interpretation of the type $\cdot \vdash S[\alpha
  \coloneqq T]$ should be given by the composition $\sem{S} \circ \sem{T} : \mathsf{1} \to \X$, which should be a face map- and degeneracy-preserving
functor. While preservation of face maps is easy to prove,
preservation of degeneracies poses a problem: writing $S_0$ and $S_1$ for $\sem{S}(0)$ and $\sem{S}(1)$, and similarly for $T$, we need $S_1(T_1)$ to be
isomorphic to the degeneracy $\dm(S_0(T_0))$. By assumption, $T_1$ is
isomorphic to the degeneracy $\dm(T_0)$, and $S_1(\dm(T_0))$ is
isomorphic to $\dm(S_0(T_0))$, so if we knew that $S_1$ mapped
isomorphic relations to isomorphic relations we would be done.
But since the domain of $S_1$ is $|\X(1)|$, there is no reason
that it should preserve non-identity isomorphisms of $\X(1)$.

In this paper we solve this contravariance problem in a different way.
We first note that the issue does not arise if $\sem{A}(0) \,f$ is an
isomorphism, even if that isomorphism is not the identity. This leads
us to require, for each $l \in \{0,1\}$, a wide subcategory $\M(l)
\subseteq \X(l)$ such that every morphism in $\M(l)$ is in fact an
isomorphism.

\begin{definition}\label{def:subrg}
Given a reflexive graph category $\X$, a \emph{reflexive graph
  subcategory} of $\X$ is a reflexive graph category $\M$ together
with a reflexive graph ``inclusion" functor $\I : \M \to \X$ such that
\begin{itemize}
\item $\I(l)_0$ and $\I(l)_1$ are monic for $l \in \{0,1\}$
\item $\I(0) \circ \M(\fm_\star) = \X(\fm_\star) \circ \I(1)$ for
  $\star \in \two$
\item $\I(1) \circ \M(\dm) = \X(\dm) \circ \I(1)$
\end{itemize}
The subcategory $(\M, \I)$ is \emph{wide} if $\I(l)_0$ is an isomorphism for $l \in \{0,1\}$.
\end{definition}
The last two conditions in Definition~\ref{def:subrg} guarantee that $\I$ preserves face maps and degeneracies on the nose. To simplify the presentation, we treat $\M(l)$ as a subcategory of $\X(l)$ and avoid explicit mentions of $\I$ unless otherwise indicated.

\begin{definition}\label{def:rgwi}
A {\em reflexive graph category with isomorphisms} is a reflexive
graph category $\X$ together with a wide reflexive graph subcategory
$(\M, \I)$ such that every morphism in $\M(l)$, $l \in \{0,1\}$, is
an isomorphism.
\end{definition}

We view $\M(l)$ as selecting the \emph{relevant isomorphisms} of $\X(l)$, in
the sense that a morphism of $\X(l)$ is relevant iff it lies in the image
of $\I(l)$. Given a reflexive graph category with isomorphisms $(\X,
(\M, \I))$ we can now interpret a type $\overline{\alpha} \vdash T$
with $n$ free type variables as a reflexive graph functor $\sem{T} :
\M^n \to \M$. It is important that $\sem{T}$ carries (tuples of) relevant
isomorphisms to relevant isomorphisms: if $\sem{T}$ were instead a functor
from $\M^n$ to $\X$, then it would not be possible to define
substitution (see Definition~\ref{def:subst}).

A trivial choice is to take $\M \coloneqq |\X|$. Then $\sem{T} :
|\X|^n \to |\X|$ and $\varepsilon_{\sem{T}}$ is necessarily the
identity natural transformation, so $\sem{T}$ preserves degeneracies
on the nose. This instantiation shows that, despite being motivated by
Reynolds' model, for which the Identity Extension Lemma holds only
up to isomorphism, our framework can also uniformly subsume strict
models of parametricity, for which the Identity Extension Lemma
holds on the nose.

\begin{example}[PER model, continued]\label{ex:per1}
We take $\M \coloneqq |\mathcal{R}_\mathit{PER}|$.
\end{example}

\begin{example}[A categorical version of Reynolds' model]\label{ex:crey1}
For each $l$, we take the objects of $\M(l)$ to be the objects of
$\mathcal{R}_\mathit{REY}(l)$, and the morphisms of $\M(l)$ to be
\emph{all} isomorphisms of $\mathcal{R}_\mathit{REY}(l)$. For example, the morphisms of $\M(0)$ are
\[\begin{array}{l}
\{(i,j) : \Set_1 \times \Set_1 \; \&\\
\hspace*{0.55in}i_\mathsf{d} = j_\mathsf{c} \,\times \, i_\mathsf{c} = j_\mathsf{d} \, \times \, j \circ i = \mathsf{id} \, \times \, i \circ j = \mathsf{id}
 \}
\end{array}\]
Here and at several places below we write $a = b$ for $\Id(a,b)$ and
$\{x : A \; \& \; B(x)\}$ for $\Sigma_{x:A}B(x)$ to enhance
readability. Moreover, $\circ$ and $\mathsf{id}$ are composition and
identity in the category $\Set$, and we use the subscripts
$(\cdot)_\mathsf{d}$ and $(\cdot)_\mathsf{c}$ to denote the domain and
codomain of a morphism. The first (or second) projection gives the
required mono from $\M(0)$ to $\Set_1$. We denote the resulting
reflexive graph category with isomorphisms by
$\mathcal{R}_\mathit{CREY}$.
\end{example}

\begin{example}[Reynolds' model, continued]\label{ex:rey1}
As mentioned in the introduction, to push the constructions through it
is sufficient to require preservation of isomorphisms on the relation
level only. This means that on the set level, we can take the relevant
isomorphisms to be just the identities, \emph{i.e.}, $\M(0) \coloneqq
|\mathcal{R}_\mathit{REY}(0)|$. On the relation level, we take the
objects of $\M(1)$ to be the objects of $\mathcal{R}_\mathit{REY}(1)$
--- {\em i.e.}, we take all relations --- and the morphisms of $\M(1)$
to be those isomorphisms of $\mathcal{R}_\mathit{REY}(1)$ whose images
under the two face maps are identities (this last condition is
necessary since face maps must preserve relevant
isomorphisms). Specifically, the morphisms of $\M(1)$ are
\[\begin{array}{l}
\{(i,j) : \mathsf{R}_1 \times \mathsf{R}_1 \; \&\\
\hspace*{0.55in}i_\mathsf{d} = j_\mathsf{c}
\,\times \, i_\mathsf{c} = j_\mathsf{d} \, \times \, j \circ i = \mathsf{id} \, \times \, i \circ j = \mathsf{id} \, \times \, i_\top = \mathsf{id} \, \times \, i_\bot = \mathsf{id}  
 \}
\end{array}\]
Here, we use the subscripts $(\cdot)_\mathsf{\top}$ and
$(\cdot)_\mathsf{\bot}$ to denote the image of a morphism in
$\mathsf{R}_1$ under the corresponding face map.
\end{example}

With this infrastructure in place we can now interpret a term
$\overline{\alpha}; x:S \vdash t : T$ as a natural transformation from
$\I \circ \sem{S}$ to $\I \circ \sem{T}$. Importantly, the components
of such a natural transformation are drawn from $\X(l)$ (as witnessed
by post-composition with $\I$), rather than just $\M(l)$, as would be
the case if we interpreted $t$ as a natural transformation from
$\sem{S}$ to $\sem{T}$. In fact, this latter interpretation would not
even be sensible, since not every term gives rise to an isomorphism
(most do not).

\section{Cartesian Closed Reflexive Graph Categories With Isomorphisms}

We want to interpret a type context of length $n$ as the natural number $n$,
types with $n$ free type variables as reflexive graph functors from
$\M^n$ to $\M$, and terms with $n$ free type variables as natural
transformations between reflexive graph functors with codomain
$\X$. Following the standard procedure, we first define, for each $n$,
a category $\M^n \to \M$ to interpret expressions with $n$ free type
variables, and then combine these categories using the usual
Grothendieck construction. This gives a fibration whose fiber over $n$
is $\M^n \to \M$.

\begin{definition}\label{def:mn-to-m}
The category $\M^n \to \M$ is defined as follows:
\begin{itemize}
\item the objects are face map- and degeneracy-preserving reflexive graph functors from $\M^n$ to $\M$
\item the morphisms from $\F$ to $\G$ are the face map- and degeneracy-preserving reflexive graph natural transformations from $\I \circ \F$ to $\I \circ \G$
\end{itemize}
\end{definition}

If $\F$ and $\Gcal$ are degeneracy-preserving then $\I \circ \F$ and
$\I \circ \Gcal$ are as well, and it is therefore sensible to require
natural transformations between the latter two to be
degeneracy-preserving. To move between the fibers we need a notion of
substitution:

\begin{definition}\label{def:subst}
For any $m$-tuple $\mathbf{F}\coloneqq (\F_0,\ldots,\F_{m-1})$ of objects in $\M^n \to \M$, the functor $\mathbf{F}^*$ from $\M^m \to \M$ to $\M^n \to \M$ is defined by $\mathbf{F}^*(\G) \coloneqq \G \circ
\langle \F_0, \ldots, \F_{m-1}\rangle$ for objects and $\mathbf{F}^*(\eta) \coloneqq \eta \circ \langle \F_0, \ldots, \F_{m-1}\rangle$ for morphisms.
\end{definition}

When giving a categorical interpretation of System F, a category for interpreting type contexts is also required. Writing $\R$ for the tuple $(\X,(\M,\I))$, we define:

\begin{definition}\label{def:contexts}
The {\em category of contexts} $\ctx(\R)$ is given by:
\begin{itemize}
\item objects are natural numbers
\item morphisms from $n$ to $m$ are $m$-tuples of objects in $\M^n \to \M$
\item the identity $\id_n : n \to n$ has as its $i^\mathit{th}$   component the $i^\mathit{th}$ projection functor $\mathsf{pr}^n_i$
\item given morphisms $\mathbf{F} : n \to m$ and $\mathbf{G} = (\G_0,\ldots,\G_{k-1}) : m \to k$, the $i^\mathit{th}$ component of the composition $\mathbf{G} \circ \mathbf{F} : n \to k$ is $\mathbf{F}^*(G_i)$
\end{itemize}
\end{definition}
\noindent That this is indeed a category follows from the lemma below:

\begin{lemma}\label{lem:subst_prop}
We have the following:
\begin{enumerate}
\item[i)] For any morphism $\mathbf{F} = (\F_0, \ldots, \F_{m-1}) : n \to m$ in $\ctx(\rel)$ and $0 \leq i < m$, we have $\mathbf{F}^\star(\mathsf{pr}^m_i) = \F_i$.
\item[ii)] For any natural number $n$, $(\mathbf{1}_n)^\star$ is the identity functor on $|\rel|^n \to \rel$.
\item[iii)] For morphisms $\mathbf{F} : n \to m$, $\mathbf{G} : m \to k$ in $\ctx(\rel)$, we have $(\mathbf{G} \circ \mathbf{F})^\star = \mathbf{F}^\star \circ \mathbf{G}^\star$.
\end{enumerate}
\end{lemma}

\begin{proof}
Parts $i)$, $ii)$ are easy to show. For part $iii)$ let $\mathbf{F} = (\F_0, \ldots, \F_{m-1})$ and $\mathbf{G} = (\G_0, \ldots, \G_{k-1})$. Fix an object $\mathcal{H}$ in $\M^k \to \M$. The first component of $(\mathbf{G} \circ \mathbf{F})^\star(\mathcal{H})$ is the reflexive graph functor whose component at level $l$ is
\begin{align*}
\mathcal{H}(l) \circ \Big\langle \G_0(l) \circ \langle \F_0(l), \ldots, \F_{m-1}(l) \rangle, \ldots, \G_{k-1}(l) \circ \langle \F_0(l), \ldots, \F_{m-1}(l)\rangle \Big\rangle
\end{align*}
On the other hand, the first component of $\mathbf{F}^\star(\mathbf{G}^\star(\mathcal{H}))$ is a reflexive graph functor whose component at level $l$ is
\begin{align*}
\mathcal{H}(l) \circ \langle \G_0(l), \ldots, \G_{k-1}(l) \rangle \circ \langle\F_0(l), \ldots, \F_{m-1}(l) \rangle
\end{align*}
which is clearly equal to the above. The second component of $(\mathbf{G} \circ \mathbf{F})^\star(\mathcal{H})$ is the morphism

\resizebox{.95\linewidth}{!}{
\begin{minipage}{\linewidth}
\begin{align*}
& \Big(\mathcal{H}(1)_1 \circ \Big\langle \big(\G_0(1)_1 \circ \langle \varepsilon_{\F_0}, \ldots, \varepsilon_{\F_{m-1}}\rangle\big) \circ_{\M(1)} \big(\varepsilon_{\G_0} \circ \langle \F_0(0)_0, \ldots, \F_{m-1}(0)_0\rangle\big), \ldots,  \\
& \;\;\;\;\;\;\;\;\; \big(\G_{k-1}(1)_1 \circ \langle \varepsilon_{\F_0}, \ldots, \varepsilon_{\F_{m-1}}\rangle\big) \circ_{\M(1)} \big(\varepsilon_{\G_{k-1}} \circ \langle \F_0(0)_0, \ldots, \F_{m-1}(0)_0\rangle\big)\Big\rangle\Big) \; \circ_{\M(1)} \\
& \;\;\; \Big(\varepsilon_\mathcal{H} \circ \Big\langle \G_0(0)_0 \circ \langle \F_0(0)_0, \ldots, \F_{m-1}(0)_0\rangle, \ldots, \\ 
& \;\;\;\;\;\;\;\;\;\;\;\;\;\;\;\;\;\;\;\;  \G_{k-1}(0)_0 \circ \langle \F_0(0)_0, \ldots, \F_{m-1}(0)_0\rangle\Big\rangle\Big)
\end{align*}
\end{minipage}}\bigskip

\noindent On the other hand, the second component of $\mathbf{F}^\star(\mathbf{G}^\star(\mathcal{H}))$ is the morphism

\resizebox{.95\linewidth}{!}{
  \begin{minipage}{\linewidth}
\begin{align*}
& \mathcal{H}(1)_1\Big(\Big\langle \G_0(1)_1(\langle \varepsilon_{\F_0}, \ldots, \varepsilon_{\F_{m-1}}\rangle), \ldots,
\G_{k-1}(1)_1(\langle \varepsilon_{\F_0}, \ldots, \varepsilon_{\F_{m-1}}\rangle)\Big\rangle\Big) \; \circ_{\M(1)} \\
& \;\;\; \bigg(\Big(\mathcal{H}(1)_1(\langle \varepsilon_{\G_0}, \ldots, \varepsilon_{\G_{k-1}}\rangle) \circ_{\M(1)} 
\big(\varepsilon_\mathcal{H} \circ \langle \G_0(0)_0, \ldots, \G_{k-1}(0)_0 \rangle\big)\Big) \; \circ \\ 
& \;\;\;\;\;\;\;\;\; \langle \F_0(0)_0, \ldots, \F_{m-1}(0)_0\rangle\bigg)
\end{align*}
\end{minipage}}\bigskip

We now have the chain of equalities in Figure~\ref{figEE}, where the first equality follows by definition of $\circ_{\M(1)^m}$; the second one follows by functoriality of $\mathcal{H}(1)$; and the third one follows since $\circ_{\M(1)}$ commutes with precomposition in the ambient category. We use the same color to denote rewriting of equal subexpressions. This finishes the proof that $(\mathbf{G} \circ \mathbf{F})^\star$ and $\mathbf{F}^\star \circ \mathbf{G}^\star$ agree on objects. The proof that they agree on morphisms is easy.
\end{proof}

\begin{figure*}[t!]
\resizebox{.95\linewidth}{!}{
\begin{minipage}{\linewidth}
\begin{align*}
& \Big(\mathcal{H}(1)_1 \circ {\color{red}\Big\langle \big(\G_0(1)_1 \circ \langle \varepsilon_{\F_0}, \ldots, \varepsilon_{\F_{m-1}}\rangle\big) \circ_{\M(1)} \big(\varepsilon_{\G_0} \circ \langle \F_0(0)_0, \ldots, \F_{m-1}(0)_0\rangle\big), \ldots,}  \\
& {\color{red}\;\;\;\;\;\;\;\;\; \big(\G_{k-1}(1)_1 \circ \langle \varepsilon_{\F_0}, \ldots, \varepsilon_{\F_{m-1}}\rangle\big) \circ_{\M(1)} \big(\varepsilon_{\G_{k-1}} \circ \langle \F_0(0)_0, \ldots, \F_{m-1}(0)_0\rangle\big)\Big\rangle}\Big) \; \circ_{\M(1)} \\
& \;\;\; \Big(\varepsilon_\mathcal{H} \circ {\color{teal} \Big\langle \G_0(0)_0 \circ \langle \F_0(0)_0, \ldots, \F_{m-1}(0)_0\rangle, \ldots,} \\
& \;\;\;\;\;\;\;\;\;\;\;\;\;\;\;\;\;\;\;\;  {\color{teal}\G_{k-1}(0)_0 \circ \langle \F_0(0)_0, \ldots, \F_{m-1}(0)_0\rangle\Big\rangle}\Big) \\
\stackrel{(1)}{=} \; & \Big(\mathcal{H}(1)_1 \circ {\color{red}\Big(\Big\langle \G_0(1)_1 \circ \langle \varepsilon_{\F_0}, \ldots, \varepsilon_{\F_{m-1}}\rangle, \ldots, \G_{k-1}(1)_1 \circ \langle \varepsilon_{\F_0}, \ldots, \varepsilon_{\F_{m-1}}\rangle \Big\rangle \; \circ_{\M(1)^n}} \\
& \;\;\;\;\;\;\; {\color{red}\big(\langle\varepsilon_{\G_0}, \ldots, \varepsilon_{\G_{k-1}}\rangle \circ \langle \F_0(0)_0, \ldots, \F_{m-1}(0)_0\rangle\big)\Big)}\Big) \; \circ_{\M(1)} \\
& \;\;\; \Big(\varepsilon_\mathcal{H} \circ {\color{teal} \langle \G_0(0)_0, \ldots, \G_{k-1}(0)_0\rangle \circ \langle \F_0(0)_0, \ldots, \F_{m-1}(0)_0\rangle}\Big) \\
\stackrel{(2)}{=} \; & \Big(\mathcal{H}(1)_1 \circ \Big\langle \G_0(1)_1 \circ \langle \varepsilon_{\F_0}, \ldots, \varepsilon_{\F_{m-1}}\rangle, \ldots, \G_{k-1}(1)_1 \circ \langle \varepsilon_{\F_0}, \ldots, \varepsilon_{\F_{m-1}}\rangle \Big\rangle\Big) \; \circ_{\M(1)} \\
& \;\;\; {\color{blue} \Big(\mathcal{H}(1)_1 \circ \langle\varepsilon_{\G_0}, \ldots, \varepsilon_{\G_{k-1}}\rangle \circ \langle \F_0(0)_0, \ldots, \F_{m-1}(0)_0\rangle\Big) \; \circ_{\M(1)}} \\
& \;\;\; {\color{blue} \Big(\varepsilon_\mathcal{H} \circ \langle \G_0(0)_0, \ldots, \G_{k-1}(0)_0\rangle \circ \langle \F_0(0)_0, \ldots, \F_{m-1}(0)_0\rangle\Big)} \\
\stackrel{(3)}{=} \; & \Big(\mathcal{H}(1)_1 \circ \Big\langle \G_0(1)_1 \circ \langle \varepsilon_{\F_0}, \ldots, \varepsilon_{\F_{m-1}}\rangle, \ldots,
\G_{k-1}(1)_1 \circ \langle \varepsilon_{\F_0}, \ldots, \varepsilon_{\F_{m-1}}\rangle\Big\rangle\Big) \; \circ_{\M(1)} \\
& {\color{blue} \;\;\; \bigg(\Big(\mathcal{H}(1)_1 \circ \langle \varepsilon_{\G_0}, \ldots, \varepsilon_{\G_{k-1}}\rangle \circ_{\M(1)} 
\big(\varepsilon_\mathcal{H} \circ \langle \G_0(0)_0, \ldots, \G_{k-1}(0)_0 \rangle\big)\Big) \; \circ} \\ 
&{\color{blue} \;\;\;\;\;\;\;\;\; \langle \F_0(0)_0, \ldots, \F_{m-1}(0)_0\rangle\bigg)}
\end{align*}
\end{minipage}}
\caption{Equalities for The Proof of Lemma~\ref{lem:subst_prop}}\label{figEE}
\end{figure*}

Defining the product $n \times 1$ in $\ctx(\rel)$ to be the natural
number sum $n + 1$, we see that $\ctx(\rel)$ enjoys sufficient
structure to model the construction of System F type contexts:

\begin{lemma}\label{lem:term-and-products}
The category $\ctx(\rel)$ has a terminal object $0$ and products $(-) \times 1$.  
\end{lemma}

\begin{proof}
The product of $n$ and $1$ is $n + 1$; the first projection has as its $i$-th component the ``$i$-th projection functor" $\mathsf{pr}^{n+1}_i$ and the second projection has as its sole component the ``$n$-th projection functor'' $\mathsf{pr}^{n+1}_{n}$.
\end{proof}

\noindent The categories $\ctx(\R)$ and $\M^n \to \M$ can be combined to give:

\begin{definition}\label{def:int-m}
The category $\int_n \, \M^n \to \M$ is defined as follows:
\begin{itemize}
\item objects are pairs $(n,\F)$, where $\F$ is an object in
  $\M^n \to \M$
\item morphisms from $(n,\F)$ to $(m,\G)$ are pairs
  $(\mathbf{F},\eta)$, where $\mathbf{F} : n \to m$ is a morphism in
  $\ctx(\X)$ and $\eta : \F \to \mathbf{F}^*(\G)$ is a morphism
  in $\M^n \to \M$
\item the identity on $(n,\F)$ is the pair $(\id_n,\id_\F)$, where
  $\id_n : n \to n$ is the identity in $\ctx(\R)$ and $\id_\F : \F \to
  \F$ is the identity in $\M^n \to \M$
\item the composition of two morphisms $(\mathbf{F},\eta_1) : (n,\F)
  \to (m,\G)$ and $(\mathbf{G},\eta_2) : (m,\G) \to (k,\mathcal{H})$ is the pair
  $(\mathbf{G} \circ \mathbf{F}, \mathbf{F}^*(\eta_2) \circ
  \eta_1)$, where the first composition is in $\ctx(\R)$ and the
  second composition is in $\M^n \to \M$
\end{itemize} 
\end{definition}

This is a standard (op)Grothendieck construction, and resuts in a
category whose objects can be understood as pairing a kinding context
and a typing context over it, and whose morphisms can be understood
as simultaneous substitutions.

Since the set of objects of $\M^n \to \M$ is, by definition,
(isomorphic to) the set of morphisms from $n$ to $1$ in $\ctx(\rel)$, we
have not only that $\int_n \, \M^n \to \M$ is the total category
of a fibration over $\ctx(\rel)$, but that this fibration is actually
a split fibration with a split generic object:

\begin{lemma}\label{lem:split-fib}
The forgetful functor from $\int_n \, \M^n \to \M$ to
$\ctx(\rel)$ is a split fibration with split generic object $1$.
\end{lemma}

\begin{proof}
Given any morphism $\mathbf{F} : n \to m$ in $\ctx(\rel)$ and an object $\G$ in $\M^m \to \M$, the cartesian lifting of $\mathbf{F}$ with respect to $\G$ is defined to be the morphism $(\mathbf{F},\mathsf{id}_{\mathbf{F}^*(\G)}) : (n,\mathbf{F}^*(\G)) \to (m,\G)$ in the total category $\int_n \, \M^n \to \M$. The induced reindexing functor is precisely $\mathbf{F}^*$.
\end{proof}

To appropriately interpret arrow types we need the category $\M^n \to \M$ to be cartesian closed. For this we require more structure on the underlying reflexive graph category with isomorphisms. We define:

\begin{definition}\label{def:terminal_int}
An internal category $C$ in $\C$ \emph{has a terminal object} if it comes equipped with an arrow $1_C : 1 \to C_0$ with the following universal property: 
\begin{itemize}
\item for any object $a : J \to C_0$ (with $J$ arbitrary), there is a unique morphism $!_C(a) : J \to C_1$ such that 
\begin{align*}
& \source_C[!_C(a)] = a \\
& \target_C[!_C(a)] = 1_C \; \circ \; !J
\end{align*}
\end{itemize}
\end{definition}
It is possible to show that the above definition is equivalent to the standard one given \emph{e.g.}, in Section 7.2 of~\cite{jac99}. However, the explicit version will be more useful for us.

\begin{definition}\label{def:terminal2}
A reflexive graph category $\X$ \emph{has terminal objects} if for each $l \in \{0,1\}$ the category $\X(l)$ has a terminal object. The terminal objects are \emph{stable under face maps} if for any $\star \in \two$, the canonical morphism witnessing the commutativity of the diagram below is the identity:

\begin{center}
\scalebox{0.9}{
\begin{tikzpicture}
\node (N0) at (0,1) {$1$};
\node (N1) at (3,2) {$\X(1)_0$};
\node (N2) at (3,0) {$\X(0)_0$};
\draw[->] (N0) -- node[above]{$1_{\X(1)\;\;\;\;}$} (N1);
\draw[->] (N0) -- node[below]{$1_{\X(0)}$} (N2);
\draw[->] (N1) -- node[right]{$\X(\fm_\star)_0$} (N2);
\end{tikzpicture}}
\end{center}
The terminal objects are \emph{stable under degeneracies} if the canonical morphism $\eta_\X^\mathsf{1}$ witnessing the commutativity of the diagram below is an isomorphism:
\begin{center}
\scalebox{0.9}{
\begin{tikzpicture}
\node (N0) at (0,1) {$1$};
\node (N1) at (3,2) {$\X(0)_0$};
\node (N2) at (3,0) {$\X(1)_0$};
\draw[->] (N0) -- node[above]{$1_{\X(0)}$} (N1);
\draw[->] (N0) -- node[below]{$1_{\X(1)}$} (N2);
\draw[->] (N1) -- node[right]{$\X(\dm)_0$} (N2);
\end{tikzpicture}}
\end{center}
\end{definition}

\begin{definition}\label{def:terminal1}
A reflexive graph category $(\X,(\M,\I))$ with isomorphisms \emph{has terminal objects} if $\X$ has terminal objects. The terminal objects are \emph{stable under face maps} if the terminal objects in $\X$ are stable under face maps. The terminal objects are \emph{stable under degeneracies} if the terminal objects in $\X$ are stable under degeneracies and the (iso)morphism $\eta_\X^\mathsf{1}$ is in the image of $\I(1)$.
\end{definition}

\begin{lemma}\label{lem:terminal-choice}
If a reflexive graph category $(\X,(\M,\I))$ with isomorphisms has terminal objects stable under face maps and degeneracies, then for each $n$, the category $\M^n \to \M$ has a terminal object.
\end{lemma}
\begin{proof}
We define the terminal object in $\M^n \to \M$ to be $1_n$, where
\begin{itemize}
\item $1_n(l)_0 \coloneqq 1_{\X(l)} \; \circ \; !(\M(l)^n_0)$
\item $1_n(l)_1 \coloneqq \id_{\M(l)}[1_{\X(l)}] \; \circ \; !(\M(l)^n_1)$
\item $\varepsilon_{1_n} \coloneqq \eta^\mathsf{1}_\X \; \circ \; !(\M(0)^n_0)$
\end{itemize}
To show that $1_n$ is indeed a terminal object, take another object $\F$. The universal morphism from $\F$ into our candidate terminal object is the reflexive graph natural transformation whose component at level $l$ is $!_{\X(l)}(\F(l)_0)$. To prove naturality, we need to show that
\[ \big(!_{\X(l)}(\F(l)_0) \circ \target_{\M(l)^n}\big) \circ_{\X(l)} \F(l)_1 = {1_n}(l)_1 \circ_{\X(l)} \big(!_{\X(l)}(\F(l)_0) \circ \source_{\M(l)^n}\big) \]
The target of both sides is $1_{\X(l)}\;\circ\;!(\M(l)^n_1)$ so the equality follows from the universal property of $1_{\X(l)}$. To prove that the candidate universal morphism is degeneracy-preserving, we need to show that
\[ \big(!_{\X(1)}(\F(1)_0) \circ \M(\dm)^n_0\big) \circ_{\X(1)} \varepsilon_\F = \varepsilon_{1_n} \circ_{\X(1)} \big(\X(\dm)_1 \; \circ \; !_{\X(0)}(\F(0)_0)\big)\]
The target of both sides is $1_{\X(1)} \circ\;!(\M(0)^n_0)$ so the equality again follows from the universal property of $1_{\X(1)}$. The preservation of face maps follows by the exact same argument.  This shows that our candidate universal morphism is indeed a proper morphism. Its uniqueness is obvious, once again by the universal property of $1_{\X(l)}$. 
\end{proof}

\begin{definition}\label{def:products_int}
An internal category $C$ in $\C$ \emph{has products} if it comes equipped with arrows $\times_C :C_0 \times C_0 \to C_0$ and $\fst_C, \snd_C : C_0 \times C_0 \to C_1$ such that 
\begin{align*}
& \source_C[\fst_C] = \times_C \;\; \mathit{and} \;\; \target_C[\fst_C] = \fst[C_0,C_0] \\
& \source_C[\snd_C] = \times_C \;\; \mathit{and} \;\; \target_C[\snd_C] = \snd[C_0,C_0]
\end{align*}
with the following universal property:
\begin{itemize}
\item for any objects $a,b,c : J \to C_0$ and morphisms $f,g : J \to C_1$ (with $J$ arbitrary) such that 
\begin{align*}
& \source_C[f] = c \;\; \mathit{and} \;\; \target_C[f] = a \\
& \source_C[g] = c \;\; \mathit{and} \;\; \target_C[g] = b
\end{align*}
there is a unique morphism $\langle f,g\rangle_C : J \to C_1$ such that 
\begin{align*}
& \source_C[\langle f,g\rangle_C] = c \\
& \target_C[\langle f,g\rangle_C] = a \times_C b \\
& \fst_C[a,b] \circ_{C} \langle f,g\rangle_C = f \\
& \snd_C[a,b] \circ_{C} \langle f,g\rangle_C = g
\end{align*}
\end{itemize}
where we write $a \times_C b$, $\fst_C[a,b]$, $\snd_C[a,b]$ for the arrows $\times_C \circ \langle a, b \rangle$, $\fst_C \circ \langle a, b \rangle$, $\snd_C \circ \langle a, b \rangle$.
\end{definition}

\noindent If $C$ has products, then we have the following:
\begin{itemize}
\item for any objects $a,b,c,d : J \to C_0$ and morphisms $f,g : J \to C_1$ such that 
\begin{align*}
& \source_C[f] = a \;\; \mathit{and} \;\; \target_C[f] = c \\
& \source_C[g] = b \;\; \mathit{and} \;\; \target_C[g] = d
\end{align*}
there exists a unique morphism $f \times_C g : J \to C_1$ such that
\begin{align*}
& \source_C[f \times_C g] = a \times_C b \\
& \target_C[f \times_C g] = c \times_C d \\
& \fst_C[c,d] \circ_{C} (f \times_C g) = f \circ_C \fst_C[a,b] \\
& \snd_C[c,d] \circ_{C} (f \times_C g) = g \circ_C \snd_C[a,b]
\end{align*}
\end{itemize}
Using this observation, it is possible to show that above definition is equivalent to the standard one given \emph{e.g.}, in Section 7.2 of~\cite{jac99}.

\begin{definition}\label{def:products2}
A reflexive graph category $\X$ \emph{has products} if for each $l \in \{0,1\}$ the category $\X(l)$ has products. The products are \emph{stable under face maps} if for any $\star \in \two$, the canonical morphism witnessing the commutativity of the diagram below is the identity:

\begin{center}
\scalebox{0.9}{
\begin{tikzpicture}
\node (N0) at (0,2) {$\X(1)_0 \times \X(1)_0$};
\node (N1) at (0,0) {$\X(0)_0 \times \X(0)_0$};
\node (N2) at (4,2) {$\X(1)_0$};
\node (N3) at (4,0) {$\X(0)_0$};
\draw[->] (N0) -- node[above]{$\times_{\X(1)}$} (N2);
\draw[->] (N1) -- node[below]{$\times_{\X(0)}$} (N3);
\draw[->] (N0) -- node[left]{$\X(\fm_\star)_0 \times \X(\fm_\star)_0$} (N1);
\draw[->] (N2) -- node[right]{$\X(\fm_\star)_0$} (N3);
\end{tikzpicture}}
\end{center}
The products are \emph{stable under degeneracies} if the canonical morphism $\eta_\X^\times$ witnessing the commutativity of the diagram below is an isomorphism:

\begin{center}
\scalebox{0.9}{
\begin{tikzpicture}
\node (N0) at (0,2) {$\X(0)_0 \times \X(0)_0$};
\node (N1) at (0,0) {$\X(1)_0 \times \X(1)_0$};
\node (N2) at (4,2) {$\X(0)_0$};
\node (N3) at (4,0) {$\X(1)_0$};
\draw[->] (N0) -- node[above]{$\times_{\X(0)}$} (N2);
\draw[->] (N1) -- node[below]{$\times_{\X(1)}$} (N3);
\draw[->] (N0) -- node[left]{$\X(\dm)_0 \times \X(\dm)_0$} (N1);
\draw[->] (N2) -- node[right]{$\X(\dm)_0$} (N3);
\end{tikzpicture}}
\end{center}
\end{definition}

\begin{notation}
If $\X$ has products stable under degeneracies, we write $\eta_\X^\times[a,b]$ for the composition $\eta_\X^\times \circ \langle a,b \rangle$ whenever $a,b : J \to \X(0)_0$ are two objects.
\end{notation}

\noindent If $\X$ has products stable under degeneracies, we have:
\begin{itemize}
\item for any objects $a,b : J \to \X(0)_0$, the following diagrams commute:

\begin{center}
\scalebox{0.75}{
\begin{tikzpicture}
\node (N0) at (0,3) {$\X(\dm)_0 \circ (a \times_{\X(1)} b)$};
\node (N1) at (7,0) {$\X(\dm)_0 \circ a$};
\node (N2) at (7,3) {$(\X(\dm)_0 \circ a) \times_{\X(1)} (\X(\dm)_0 \circ a)$};
\draw[->] (N0) -- node[below]{$\X(\dm)_1 \circ \fst_{\X(0)}[a,b]$\;\;\;\;\;\;\;\;\;\;\;\;\;\;\;\;\;\;\;\;\;\;\;\;\;\;\;\;} (N1);
\draw[->] (N0) -- node[above]{$\eta_\X^\times[a,b]$} (N2);
\draw[->] (N2) -- node[right]{$\fst_{\X(1)}\big[\X(\dm)_0 \circ a, \X(\dm)_0 \circ b\big]$} (N1);
\end{tikzpicture}}
\end{center}

\begin{center}
\scalebox{0.75}{
\begin{tikzpicture}
\node (N0) at (0,3) {$\X(\dm)_0 \circ (a \times_{\X(1)} b)$};
\node (N1) at (7,0) {$\X(\dm)_0 \circ a$};
\node (N2) at (7,3) {$(\X(\dm)_0 \circ a) \times_{\X(1)} (\X(\dm)_0 \circ a)$};
\draw[->] (N0) -- node[below]{$\X(\dm)_1 \circ \snd_{\X(0)}[a,b]$\;\;\;\;\;\;\;\;\;\;\;\;\;\;\;\;\;\;\;\;\;\;\;\;\;\;\;\;} (N1);
\draw[->] (N0) -- node[above]{$\eta_\X^\times[a,b]$} (N2);
\draw[->] (N2) -- node[right]{$\snd_{\X(1)}\big[\X(\dm)_0 \circ a, \X(\dm)_0 \circ b\big]$} (N1);
\end{tikzpicture}}
\end{center}

\item The isomorphism $\eta_\X^\times$ is coherent, \emph{i.e.}, for any objects $a,b : J \to \X(0)_0$:
\[ \X(\fm_\star)_1 \circ \eta_\X^\times[a,b] = \mathsf{id} \]
\item The isomorphism $\eta_\X^\times$ is natural, \emph{i.e.}, for any objects $a,b,c,d : J \to \X(0)_0$ and morphisms $f,g : J \to \X(0)_1$ such that 
\begin{align*}
& \source_{\X(0)}[f] = a \;\; \mathit{and} \;\; \target_{\X(0)}[f] = c \\
& \source_{\X(0)}[g] = b \;\; \mathit{and} \;\; \target_{\X(0)}[g] = d
\end{align*}
the following diagram commutes:
\begin{center}
\scalebox{0.75}{
\begin{tikzpicture}
\node (N0) at (0,2) {$\X(\dm)_0 \circ (a \times_{\X(0)} b)$};
\node (N1) at (0,0) {$(\X(\dm)_0 \circ a) \times_{\X(1)} (\X(\dm)_0 \circ b)$};
\node (N2) at (9.5,2) {$\X(\dm)_0 \circ (c \times_{\X(0)} d)$};
\node (N3) at (9.5,0) {$(\X(\dm)_0 \circ c) \times_{\X(1)} (\X(\dm)_0 \circ d)$};
\draw[->] (N0) -- node[left]{$\eta_\X^\times[a,b]$} (N1);
\draw[->] (N0) -- node[above]{$\X(\dm)_1 \circ (f \times_{\X(0)} g)$} (N2);
\draw[->] (N1) -- node[below]{$(\X(\dm)_1 \circ f) \times_{\X(1)} (\X(\dm)_1 \circ g)$} (N3);
\draw[->] (N2) -- node[right]{$\eta_\X^\times[c,d]$} (N3);
\end{tikzpicture}}
\end{center}
\end{itemize}

\begin{definition}\label{def:products1}
A reflexive graph category $(\X,(\M,\I))$ with isomorphisms \emph{has products} if $\X$ has products and for any $f,g : J \to \X(l)_1$, $f \times_{\X(l)} g$ is in the image of $\I(l)$ whenever $f$ and $g$ are. The products are \emph{stable under face maps} if the products in $\X$ are stable under face maps. The products are \emph{stable under degeneracies} if the products in $\X$ are stable under degeneracies and the (iso)morphism $\eta_\X^\times$ is in the image of $\I(1)$.
\end{definition}

\begin{lemma}\label{lem:products-choice}
If a reflexive graph category $(\X,(\M,\I))$ with isomorphisms has products stable under face maps and degeneracies, then for each $n$, the category $\M^n \to \M$ has products.
\end{lemma}
\begin{proof}
Fix $\F$ and $\G$ in $\M^n \to \M$. We define $\F \times \G$ by:
\begin{itemize}
\item $(\F \times \G)(l)_0\coloneqq \F(l)_0 \times_{\X(l)} \G(l)_0$
\item $(\F \times \G)(l)_1\coloneqq \F(l)_1 \times_{\X(l)} \G(l)_1$
\item $\varepsilon_{\F \times \G} \coloneqq (\varepsilon_\F \times_{\X(1)} \varepsilon_\G) \circ_{\M(1)} \eta_\X^\mathsf{\times}[\F(0)_0,\G(0)_0]$
\end{itemize}
The first projection out of $\F \times \G$ is defined as the reflexive graph natural transformation whose component at level $l$ is $\fst_{\X(l)}[\F(l)_0,\G(l)_0]$. To prove naturality -- with respect to $\F \times \G$ and $\G$ -- we observe the following chain of equalities:
\begin{align*}
& \big(\fst_{\X(l)}[\F(l)_0,\G(l)_0] \circ \target_{\M(l)^n}\big) \circ_{\X(l)} (\F(l)_1 \times_{\X(l)} \G(l)_1) \\   
= \; & \fst_{\X(l)}\big[\F(l)_0 \circ \target_{\M(l)^n},\G(l)_0 \circ \target_{\M(l)^n}\big] \circ_{\X(l)} (\F(l)_1 \times_{\X(l)} \G(l)_1) \\ 
= \; & \F(l)_1 \circ_{\X(l)} \fst_{\X(l)}\big[\F(l)_0 \circ \source_{\M(l)^n},\G(l)_0 \circ \source_{\M(l)^n}\big] \\
= \; & \F(l)_1 \circ_{\X(l)} \big(\fst_{\X(l)}[\F(l)_0,\G(l)_0] \circ \source_{\M(l)^n}\big)
\end{align*}
The first and third equalities are clear and the second follows by the definition of $\times_{\X(l)}$ on morphisms. To prove degeneracy-preservation -- with respect to $\varepsilon_{\F \times \G}$ and $\varepsilon_\G$ -- we observe the following chain of equalities:

\resizebox{.95\linewidth}{!}{
\begin{minipage}{\linewidth}
\begin{align*}
& \big(\fst_{\X(1)}[\F(1)_0,\G(1)_0] \circ \M(\dm)_0^n\big) \circ_{\X(1)} \big(\varepsilon_\F \times_{\X(1)} \varepsilon_\G\big) \circ_{\X(1)} \eta_\X^\mathsf{\times}[\F(0)_0,\G(0)_0] \\
= \; & \fst_{\X(1)}\big[\F(1)_0 \circ \M(\dm)_0^n,\G(l)_0 \circ \M(\dm)_0^n\big] \circ_{\X(1)} (\varepsilon_\F \times_{\X(1)} \varepsilon_\G) \circ_{\X(1)} \eta_\X^\mathsf{\times}[\F(0)_0,\G(0)_0] \\
= \; & \varepsilon_\F \circ_{\X(1)} \fst_{\X(1)}\big[\X(\dm)_0 \circ \F(0)_0,\X(\dm)_0 \circ \G(0)_0\big] \circ_{\X(1)} \eta_\X^\mathsf{\times}[\F(0)_0,\G(0)_0] \\
= \; & \varepsilon_\F \circ_{\X(1)} \big(\X(\dm)_1 \circ \fst_{\X(0)}[\F(0)_0,\G(0)_0]\big)
\end{align*}
\end{minipage}}\bigskip

\noindent The first equality is clear, the second follows by definition of $\times_{\X(1)}$ on morphisms, and the third follows by definition of $\eta_\X^\times$. The preservation of face maps follows by the exact same argument. This shows that the first projection is indeed a proper morphism. The second projection is defined analogously. \smallskip

To show that $\F \times \G$ with the aforementioned projections is indeed a product, fix $\mathcal{H}$ and $\eta_\F : \mathcal{H} \to \F$, $\eta_\G : \mathcal{H} \to \G$. The universal morphism from $\mathcal{H}$ into $\F \times \G$ is the reflexive graph natural transformation whose component at level $l$ is $\langle \eta_\F(l),\eta_\G(l) \rangle_{\X(l)}$. To show naturality -- with respect to $\mathcal{H}$ and $\F \times \G$ -- we need to establish the equality
\begin{align*} 
& \big(\langle \eta_\F(l),\eta_\G(l) \rangle_{\X(l)} \circ \target_{\M(l)^n}\big) \circ_{\X(l)} \mathcal{H}(l)_1 = \\
& \;\;\; (\F(l)_1 \times_{\X(l)} \G(l)_1) \circ_{\X(l)} \big(\langle \eta_\F(l),\eta_\G(l) \rangle_{\X(l)} \circ \source_{\M(l)^n}\big)
\end{align*}

The target of the two morphisms is a product, so it suffices to check that their compositions with the first and second projections coincide. The chain of equalities below establishes this for the first projection. Equalities $(1)$ and $(5)$ are clear; equalities $(2)$ and $(4)$ follow by the definition of $\langle \cdot,\cdot\rangle_{\X(l)}$; equality $(3)$ follows from the naturality of $\eta_\F$; and equality $(6)$ follows by the definition of $\times_{\X(l)}$ on morphisms. The case of the second projection is entirely analogous.

\resizebox{.95\linewidth}{!}{
\begin{minipage}{\linewidth}
\begin{align*}
& \fst_{\X(l)}\big[\F(l)_0 \circ \target_{\M(l)^n}, \G(l)_0 \circ \target_{\M(l)^n}\big] \circ_{\X(l)} \big(\langle \eta_\F(l),\eta_\G(l) \rangle_{\X(l)} \circ \target_{\M(l)^n}\big) \circ_{\X(l)} \mathcal{H}(l)_1 \\
\stackrel{(1)}{=} \; & \Big(\big(\fst_{\X(l)}[\F(l)_0, \G(l)_0] \circ_{\X(l)} \langle \eta_\F(l),\eta_\G(l) \rangle_{\X(l)}\big) \circ \target_{\M(l)^n}\Big) \circ_{\X(l)} \mathcal{H}(l)_1 \\
\stackrel{(2)}{=} \; & (\eta_\F(l) \circ \target_{\M(l)^n}) \circ_{\X(l)} \mathcal{H}(l)_1 \\
\stackrel{(3)}{=} \; & \F(l)_1 \circ_{\X(l)} (\eta_\F(l) \circ \source_{\M(l)^n}) \\
\stackrel{(4)}{=} \; & \F(l)_1 \circ_{\X(l)} \Big(\big(\fst_{\X(l)}[\F(l)_0, \G(l)_0] \circ_{\X(l)} \langle \eta_\F(l),\eta_\G(l) \rangle_{\X(l)} \big) \circ \source_{\M(l)^n}\Big) \\
\stackrel{(5)}{=} \; & {\color{red}\F(l)_1 \circ_{\X(l)} \fst_{\X(l)}\big[\F(l)_0 \circ \source_{\M(l)^n}, \G(l)_0 \circ \source_{\M(l)^n}\big]} \circ_{\X(l)} \big(\langle \eta_\F(l),\eta_\G(l) \rangle_{\X(l)} \circ \source_{\M(l)^n}\big) \\
\stackrel{(6)}{=} \; & {\color{red}\fst_{\X(l)}\big[\F(l)_0 \circ \target_{\M(l)^n}, \G(l)_0 \circ \target_{\M(l)^n}\big] \circ_{\X(l)} (\F(l)_1 \times_{\X(l)} \G(l)_1)} \; \circ_{\X(l)} \\
& \;\;\; \big(\langle \eta_\F(l),\eta_\G(l) \rangle_{\X(l)} \circ \source_{\M(l)^n}\big)
\end{align*}
\end{minipage}}\bigskip

To prove that our candidate universal morphism is degeneracy-preserving -- with respect to $\varepsilon_\mathcal{H}$ and $\varepsilon_{\F \times \G}$ -- we need to establish the equality
\begin{align*}
& \big(\langle \eta_\F(1),\eta_\G(1) \rangle_{\X(1)} \circ \M(\dm)_0^n\big) \circ_{\X(1)} \varepsilon_\mathcal{H} = \\
& \;\;\; (\varepsilon_\F \times_{\X(1)} \varepsilon_\G) \circ_{\X(1)} \eta_\X^\mathsf{\times}[\F(0)_0,\G(0)_0] \circ_{\X(1)} \big(\X(\dm)_1 \circ \langle \eta_\F(0),\eta_\G(0) \rangle_{\X(0)}\big)
\end{align*}

Again the target of the two morphisms is a product so it suffices to check that their compositions with the first and second projections coincide. The chain of equalities below establishes this for the first projection. Equality $(1)$ is clear; equalities $(2)$ and $(4)$ follow by the definition of $\langle \cdot,\cdot\rangle_{\X(l)}$; equality $(3)$ follows by the degeneracy-preservation of $\eta_\F$; equality $(5)$ follows by the functoriality of $\X(\dm)$; equality $(6)$ follows by the definition of $\eta_\X^\times$; and equality $(7)$ follows by the definition of $\times_{\X(1)}$ on morphisms. The case of the second projection is entirely analogous, which shows that our candidate universal morphism is degeneracy-preserving. The preservation of face maps is shown by the exact same argument. Thus our candidate universal morphism is indeed a proper morphism. Its universality and uniqueness are obvious, again by the universal property of $\times_{\X(l)}$.

\resizebox{.95\linewidth}{!}{
\begin{minipage}{\linewidth}
\begin{align*}
& \fst_{\X(1)}\big[\F(1)_0 \circ \M(\dm)^n_0,\G(1)_0 \circ \M(\dm)^n_0] \circ_{\X(1)} \big(\langle \eta_\F(1),\eta_\G(1) \rangle_{\X(1)} \circ \M(\dm)_0^n\big) \circ_{\X(1)} \varepsilon_\mathcal{H} \\
\stackrel{(1)}{=} \; &  \Big(\big(\fst_{\X(1)}[\F(1)_0,\G(1)_0] \circ_{\X(1)} \langle \eta_\F(1),\eta_\G(1) \rangle_{\X(1)}\big) \circ \M(\dm)_0^n\Big) \circ_{\X(1)} \varepsilon_\mathcal{H} \\
\stackrel{(2)}{=} \; &  (\eta_\F(1) \circ \M(\dm)_0^n) \circ_{\X(1)} \varepsilon_\mathcal{H} \\
\stackrel{(3)}{=} \; & \varepsilon_\F \circ_{\X(1)} (\X(\dm)_1 \circ \eta_\F(0)) \\
\stackrel{(4)}{=} \; & \varepsilon_\F \circ_{\X(1)} \Big(\X(\dm)_1 \circ \big(\fst_{\X(0)}[\F(0)_0,\G(0)_0] \circ_{\X(0)} \langle \eta_\F(0),\eta_\G(0) \rangle_{\X(0)}\big)\Big) \\
\stackrel{(5)}{=} \; & \varepsilon_\F \circ_{\X(1)} {\color{red}\big(\X(\dm)_1 \circ \fst_{\X(0)}[\F(0)_0,\G(0)_0]\big)} \circ_{\X(1)} \big(\X(\dm)_1 \circ \langle \eta_\F(0),\eta_\G(0) \rangle_{\X(0)}\big) \\
\stackrel{(6)}{=} \; & \varepsilon_\F \circ_{\X(1)} {\color{red}\fst_{\X(1)}\big[\X(\dm)_0 \circ \F(0)_0,\X(\dm)_0 \circ \G(0)_0] \circ_{\X(1)} \eta_\X^\mathsf{\times}[\F(0)_0,\G(0)_0]} \; \circ_{\X(1)}
\\ & \;\;\; \big(\X(\dm)_1 \circ \langle \eta_\F(0),\eta_\G(0) \rangle_{\X(0)}\big) \\
\stackrel{(7)}{=} \; & \fst_{\X(1)}\big[\F(1)_0 \circ \M(\dm)^n_0,\G(1)_0 \circ \M(\dm)^n_0] \circ_{\X(1)} (\varepsilon_\F \times_{\X(1)} \varepsilon_\G) \circ_{\X(1)} \eta_\X^\mathsf{\times}[\F(0)_0,\G(0)_0] \; \circ_{\X(1)} \\ & \;\;\; \big(\X(\dm)_1 \circ \langle \eta_\F(0),\eta_\G(0) \rangle_{\X(0)}\big)
\end{align*}
\end{minipage}}\bigskip

\end{proof}

\begin{definition}\label{def:exponentials_int}
An internal category $C$ in $\C$ with products \emph{has exponentials} if it comes equipped with arrows $\Rightarrow_C \; : C_0 \times C_0 \to C_0$ and $\eval_C : C_0 \times C_0 \to C_1$ such that 
\begin{align*}
\source_C[\eval_C] = \big((\Rightarrow_C) \times_C \fst[C_0,C_0] \big)\;\; \mathit{and} \;\; \target_C[\eval_C] = \snd[C_0,C_0]
\end{align*}
with the following universal property:
\begin{itemize}
\item for any objects $a,b,c : J \to C_0$ and morphism $f : J \to C_1$ (with $J$ arbitrary) such that 
\begin{align*}
& \source_C[f] = c \times_C a \;\; \mathit{and} \;\; \target_C[f] = b
\end{align*}
there is a unique morphism $\lambda_C[a,b,c,f] : J \to C_1$ such that 
\begin{align*}
& \source_C[\lambda_C[a,b,c,f]] = c \\
& \target_C[\lambda_C[a,b,c,f]] = (a \Rightarrow_C b) \\
& \eval_C[a,b] \circ_C \big(\lambda_C[a,b,c,f] \times_C \id_C[a]\big) = f
\end{align*}
\end{itemize}
where we write $a \Rightarrow_C b$, $\eval_C[a,b]$ for the arrows $(\Rightarrow_C \circ\; \langle a, b \rangle)$, $\eval_C \circ \langle a, b \rangle$.
\end{definition}

\noindent If $C$ has exponentials, then we have the following:
\begin{itemize}
\item for any objects $a,b,c,d : J \to C_0$ and morphisms $f,g : J \to C_1$ such that 
\begin{align*}
& \source_C[f] = c \;\; \mathit{and} \;\; \target_C[f] = a \\
& \source_C[g] = b \;\; \mathit{and} \;\; \target_C[g] = d
\end{align*}
there exists a unique morphism $f \Rightarrow_C g : J \to C_1$ such that
\begin{align*}
& \source_C[f \Rightarrow_C g] = (a \Rightarrow_C b) \\
& \target_C[f \Rightarrow_C g] = (c \Rightarrow_C d) \\
& \eval_C[c,d] \circ_{C} \big((f \Rightarrow_C g) \times_C \id_C[c]\big) =  \\
& \;\;\; g \circ_C \eval_C[a,b] \circ_C (\id_C[a \Rightarrow_C b] \times_C f)  
\end{align*}
\end{itemize}
Using this observation, it is possible to show that above definition is equivalent to the standard one given \emph{e.g.}, in Section 7.2 of~\cite{jac99}.

\begin{definition}\label{def:exponentials2}
A reflexive graph category $\X$ with products \emph{has exponentials} if for each $l \in \{0,1\}$, the category $\X(l)$ has exponentials. Assuming the products are stable under face maps, we say the exponentials are \emph{stable under face maps} if for any $\star \in \two$, the canonical morphism witnessing the commutativity of the diagram below is the identity:
	
\begin{center}
\scalebox{0.9}{
\begin{tikzpicture}
\node (N0) at (0,2) {$\X(1)_0 \times \X(1)_0$};
\node (N1) at (0,0) {$\X(0)_0 \times \X(0)_0$};
\node (N2) at (4,2) {$\X(1)_0$};
\node (N3) at (4,0) {$\X(0)_0$};
\draw[->] (N0) -- node[above]{$\Rightarrow_{\X(1)}$} (N2);
\draw[->] (N1) -- node[below]{$\Rightarrow_{\X(0)}$} (N3);
\draw[->] (N0) -- node[left]{$\X(\fm_\star)_0 \times \X(\fm_\star)_0$} (N1);
\draw[->] (N2) -- node[right]{$\X(\fm_\star)_0$} (N3);
\end{tikzpicture}}
\end{center}
Assuming the products are stable under degeneracies, we say the exponentials are \emph{stable under degeneracies} if the canonical morphism $\eta^\Rightarrow_\X$ witnessing the commutativity of the diagram below is an isomorphism:

\begin{center}
\scalebox{0.9}{
\begin{tikzpicture}
\node (N0) at (0,2) {$\X(0)_0 \times \X(0)_0$};
\node (N1) at (0,0) {$\X(1)_0 \times \X(1)_0$};
\node (N2) at (4,2) {$\X(0)_0$};
\node (N3) at (4,0) {$\X(1)_0$};
\draw[->] (N0) -- node[above]{$\Rightarrow_{\X(0)}$} (N2);
\draw[->] (N1) -- node[below]{$\Rightarrow_{\X(1)}$} (N3);
\draw[->] (N0) -- node[left]{$\X(\dm)_0 \times \X(\dm)_0$} (N1);
\draw[->] (N2) -- node[right]{$\X(\dm)_0$} (N3);
\end{tikzpicture}}
\end{center}
\end{definition}

\begin{notation}
If $\X$ has exponentials stable under degeneracies, we write $\eta_\X^\Rightarrow[a,b]$ for the composition $\eta_\X^\Rightarrow \circ \langle a,b \rangle$ whenever $a,b : J \to \X(l)_0$ are two objects.
\end{notation}

\noindent If $\X$ has products and exponentials stable under degeneracies, we have:
\begin{itemize}
\item for any objects $a,b : J \to \X(0)_0$, the following diagram commutes:

\begin{center}
\scalebox{0.75}{
\begin{tikzpicture}
\node (N0) at (0,6) {$\X(\dm)_0 \circ \big((a \Rightarrow_{\X(0)} b) \times_{\X(0)} a\big)$};
\node (N1) at (0,0) {$\X(\dm)_0 \circ b$};
\node (N2) at (9,6) {$\big(\X(\dm)_0 \circ (a \Rightarrow_{\X(1)} b)\big) \times_{\X(1)} (\X(\dm)_0 \circ a)$};
\node (N3) at (9,3) {$\;\;\;\;\;\;\;\;\;\;\;\;\;\;\;\; \big((\X(\dm) \circ a) \Rightarrow_{\X(1)} (\X(\dm) \circ b)\big) \times_{\X(1)} (\X(\dm)_0 \circ a)$};
\node at (6.5,1.1) {$\eval_{\X(1)}\big[\X(\dm)_0 \circ a, \X(\dm)_0 \circ b\big]$};
\draw[->] (N0) -- node[left]{$\X(\dm)_1 \circ \eval_{\X(0)}[a,b]$} (N1);
\draw[->] (N0) -- node[above]{$\eta_\X^\times[a \Rightarrow_{\X(0)} b,a]$} (N2);
\draw[->] (N2) -- node[right]{$\eta_\X^\Rightarrow[a,b] \times_{\X(1)} \id_{\X(1)}[\X(\dm)_0 \circ a]$} (N3);
\draw[->] (N3) -- node[below]{} (N1);
\end{tikzpicture}}
\end{center}

\item The isomorphism $\eta_\X^\Rightarrow$ is coherent, \emph{i.e.}, for any objects $a,b : J \to \X(0)_0$, the following holds:
\[ \X(\fm_\star)_1 \circ \eta_\X^\Rightarrow[a,b] = \mathsf{id} \]
\item The isomorphism $\eta_\X^\Rightarrow$ is natural, \emph{i.e.}, for any objects $a,b,c,d : J \to \X(0)_0$ and morphisms $f,g : J \to \X(0)_1$ such that 
\begin{align*}
& \source_{\X(0)}[f] = a \;\; \mathit{and} \;\; \target_{\X(0)}[f] = c \\
& \source_{\X(0)}[g] = b \;\; \mathit{and} \;\; \target_{\X(0)}[g] = d
\end{align*}
the following diagram commutes:
\begin{center}
\scalebox{0.75}{
\begin{tikzpicture}
\node (N0) at (0,2) {$\X(\dm)_0 \circ (a \Rightarrow_{\X(0)} b)$};
\node (N1) at (0,0) {$(\X(\dm)_0 \circ a) \Rightarrow_{\X(1)} (\X(\dm)_0 \circ b)$};
\node (N2) at (9.5,2) {$\X(\dm)_0 \circ (c \Rightarrow_{\X(0)} d)$};
\node (N3) at (9.5,0) {$(\X(\dm)_0 \circ c) \Rightarrow_{\X(1)} (\X(\dm)_0 \circ d)$};
\draw[->] (N0) -- node[left]{$\eta_\X^\Rightarrow[a,b]$} (N1);
\draw[->] (N0) -- node[above]{$\X(\dm)_1 \circ (f \Rightarrow_{\X(0)} g)$} (N2);
\draw[->] (N1) -- node[below]{$(\X(\dm)_1 \circ f) \Rightarrow_{\X(1)} (\X(\dm)_1 \circ g)$} (N3);
\draw[->] (N2) -- node[right]{$\eta_\X^\Rightarrow[c,d]$} (N3);
\end{tikzpicture}}
\end{center}
\end{itemize}

\begin{definition}\label{def:exponentials1}
A reflexive graph category $(\X,(\M,\I))$ with isomorphisms and products \emph{has exponentials} if $\X$ has exponentials and for any $f,g : J \to \X(l)_1$, $f \Rightarrow_{\X(l)} g$ is in the image of $\I(l)$ whenever $f$ and $g$ are. Assuming the products are stable under face maps, we say the exponentials are \emph{stable under face maps} if the exponentials in $\X$ are stable under face maps. Assuming the products are stable under degeneracies, we say the exponentials are \emph{stable under degeneracies} if the exponentials in $\X$ are stable under degeneracies and the (iso)morphism $\eta_\X^\Rightarrow$ is in the image of $\I(1)$.
\end{definition}

\begin{lemma}\label{lem:exponentials-choice}
If a reflexive graph category $(\X,(\M,\I))$ with isomorphisms has products and exponentials stable under face maps and degeneracies, then for each $n$, the category $\M^n \to \M$ has exponentials.
\end{lemma}
\begin{proof}
Fix $\F$ and $\G$ in $\M^n \to \M$. We define $\F \Rightarrow \G$ by:
\begin{itemize}
\item $(\F \Rightarrow \G)(l)_0 \coloneqq \F(l)_0 \Rightarrow_{\X(l)} \G(l)_0$
\item $(\F \Rightarrow \G)(l)_1 \coloneqq \F(l)^{-1}_1 \Rightarrow_{\X(l)} \G(l)_1$
\item $\varepsilon_{\F \Rightarrow \G} \coloneqq (\varepsilon_\F^{-1} \Rightarrow_{\X(1)} \varepsilon_\G) \circ_{\M(1)} \eta_\X^\mathsf{\Rightarrow}[\F(0)_0,\G(0)_0]$ 
\end{itemize}
The evaluation morphism for $\F \Rightarrow \G$ is defined as the reflexive graph natural transformation whose component at level $l$ is $\eval_{\X(l)}[\F(l)_0,\G(l)_0]$. To prove naturality -- with respect to $(\F \Rightarrow \G) \times \F$ and $\G$ -- we observe the chain of equalities below. The first and third equalities follow since $\times_{\X(1)}$ suitably commutes with $\circ_{\X(1)}$; the second equality follows by definition of $\Rightarrow_{\X(l)}$ on morphisms; and the fourth equality follows since the product of identities is again an identity.

\resizebox{.95\linewidth}{!}{
\begin{minipage}{\linewidth}
\begin{align*}
& {\color{teal}\big(\eval_{\X(l)}[\F(l)_0,\G(l)_0] \circ \target_{\M(l)^n}\big)} \circ_{\X(l)} {\color{red}\Big(\big(\F(l)_1^{-1} \Rightarrow_{\X(l)} \G(l)_1\big) \times_{\X(l)} \F(l)_1\Big)} \\
\stackrel{(1)}{=} \; & {\color{teal}\eval_{\X(l)}\big[\F(l)_0 \circ \target_{\M(l)^n},\G(l)_0 \circ \target_{\M(l)^n}\big]} \; \circ_{\X(l)} \\ 
& \;\;\; {\color{red}\Big(\big(\F(l)_1^{-1} \Rightarrow_{\X(l)} \G(l)_1\big) \times_{\X(l)} \id_{\X(l)}[\F(l)_0 \circ \target_{\M(l)^n}]\Big) \; \circ_{\X(l)}}\\
& \;\;\; {\color{red}\Big(\id_{\X(l)}\big[(\F(l)_0 \circ \source_{\M(l)^n}) \Rightarrow_{\X(l)} (\G(l)_0 \circ \source_{\M(l)^n})\big] \times_{\X(l)} \F(l)_1\Big)} \\
\stackrel{(2)}{=} \; & \G(l)_1 \circ_{\X(l)} \eval_{\X(l)}\big[\F(l)_0 \circ \source_{\M(l)^n},\G(l)_0 \circ \source_{\M(l)^n}\big] \; \circ_{\X(l)} \; \\
& \;\;\; {\color{blue}\Big(\id_{\X(l)}\big[(\F(l)_0 \circ \source_{\M(l)^n}) \Rightarrow_{\X(l)} (\G(l)_0 \circ \source_{\M(l)^n})\big] \times_{\X(l)} \F(l)_1^{-1}\Big) \; \circ_{\X(l)}} \\
& \;\;\; {\color{blue}\Big(\id_{\X(l)}\big[(\F(l)_0 \circ \source_{\M(l)^n}) \Rightarrow_{\X(l)} (\G(l)_0 \circ \source_{\M(l)^n})\big] \times_{\X(l)} \F(l)_1\Big)} \\
\stackrel{(3)}{=} \; & \G(l)_1 \circ_{\X(l)} \big(\eval_{\X(l)}\big[\F(l)_0,\G(l)_0] \circ \source_{\M(l)^n}\big) \; \circ_{\X(l)} \; \\
& \;\;\; {\color{blue}\Big(\id_{\X(l)}\big[(\F(l)_0 \circ \source_{\M(l)^n}) \Rightarrow_{\X(l)} (\G(l)_0 \circ \source_{\M(l)^n})\big] \times_{\X(l)} \id_{\X(l)}[\F(l)_0 \circ \source_{\M(l)^n}]\Big)} \\
\stackrel{(4)}{=} \; & \G(l)_1 \circ_{\X(l)} \big(\eval_{\X(l)}\big[\F(l)_0,\G(l)_0] \circ \source_{\M(l)^n}\big)
\end{align*}
\end{minipage}}\bigskip

To prove the degeneracy-preservation of the evaluation morphism -- with respect to $\varepsilon_{(\F \Rightarrow \G) \times \F}$ and $\varepsilon_\G$ -- we observe the chain of equalities below. 

\resizebox{.95\linewidth}{!}{
\begin{minipage}{\linewidth}
\begin{align*}
& {\color{teal}\big(\eval_{\X(1)}[\F(1)_0,\G(1)_0] \circ \M(\dm)_0^n\big)} \; \circ_{\X(1)} \\
& \;\;\; {\color{red}\bigg(\Big(\big(\varepsilon_\F^{-1} \Rightarrow_{\X(1)} \varepsilon_\G\big) \circ_{\X(1)} \eta_\X^\mathsf{\Rightarrow}[\F(0)_0,\G(0)_0]\Big) \times_{\X(1)} \varepsilon_\F\bigg)} \; \circ_{\X(1)}\\
& \;\;\;  \eta_\X^\mathsf{\times}\big[\F(0)_0 \Rightarrow_{\X(0)} \G(0)_0, \F(0)_0\big] \\
\stackrel{(1)}{=} \; & {\color{teal}\eval_{\X(1)}\big[\F(1)_0 \circ \M(\dm)_0^n,\G(1)_0 \circ \M(\dm)_0^n\big]} \; \circ_{\X(1)} \\ 
& \;\;\; {\color{red} \Big(\big(\varepsilon_\F^{-1} \Rightarrow_{\X(1)} \varepsilon_\G\big) \times_{\X(1)} \id_{\X(1)}[\F(1)_0 \circ \M(\dm)^n_0]\Big) \; \circ_{\X(1)}} \\
& \;\;\; {\color{red}\big(\eta_\X^\mathsf{\Rightarrow}[\F(0)_0,\G(0)_0] \times_{\X(1)} \varepsilon_\F\big)} \; \circ_{\X(1)} \\
& \;\;\; \eta_\X^\mathsf{\times}\big[\F(0)_0,\F(0) \Rightarrow_{\X(0)} \G(0)_0\big] \\
\stackrel{(2)}{=} \; & \varepsilon_\G \circ_{\X(1)} \eval_{\X(1)}\big[\X(\dm)_0 \circ \F(0)_0,\X(\dm)_0 \circ \G(0)_0\big] \; \circ_{\X(1)} \; \\
& \;\;\; {\color{blue}\Big(\id_{\X(1)}\big[(\X(\dm)_0 \circ \F(0)_0) \Rightarrow_{\X(1)} (\X(\dm)_0 \circ \G(0)_0)\big] \times_{\X(1)} \varepsilon_\F^{-1}\Big)} \; \circ_{l_2} \\
& \;\;\; {\color{blue}\big(\eta_\X^\mathsf{\Rightarrow}[\F(0)_0,\G(0)_0] \times_{\X(1)} \varepsilon_\F\big)} \; \circ_{\X(1)} \\
& \;\;\; \eta_\X^\mathsf{\times}\big[\F(0)_0,\F(0) \Rightarrow_{\X(0)} \G(0)_0\big] \\
\stackrel{(3)}{=} \; & \varepsilon_\G \circ_{\X(1)} \eval_{\X(1)}\big[\X(\dm)_0 \circ \F(0)_0,\X(\dm)_0 \circ \G(0)_0\big] \; \circ_{\X(1)} \\ 
& \;\;\; {\color{blue}\Big(\eta_\X^\mathsf{\Rightarrow}[\F(0)_0,\G(0)_0] \times_{\X(1)} \id_{\X(1)}[\X(\dm)_0 \circ \F(0)_0]\Big)} \; \circ_{\X(1)} \\
& \;\;\; \eta_\X^\mathsf{\times}\big[\F(0)_0,\F(0) \Rightarrow_{\X(0)} \G(0)_0\big] \\
\stackrel{(4)}{=} \; & \varepsilon_\G \circ_{\X(1)} \big(\X(\dm)_1 \circ \eval_{\X(0)}[\F(0)_0,\G(0)_0]\big)
\end{align*}
\end{minipage}}\bigskip

The first and third equalities follow since $\times_{\X(1)}$ suitably commutes with $\circ_{\X(1)}$; the second equality follows by definition of $\Rightarrow_{\X(1)}$ on morphisms; and the fourth equality follows by definition of $\eta_\X^\Rightarrow$. The preservation of face maps follows by the exact same argument. This shows that the evaluation morphism is indeed a proper morphism.\smallskip

To show that $\F \Rightarrow \G$ with the aforementioned evaluation is an exponential, fix $\mathcal{H}$ and $\eta : \F \times \mathcal{H} \to \G$. The universal morphism from $\mathcal{H}$ into $\F \Rightarrow \G$ is the reflexive graph natural transformation whose component at level $l$ is $\lambda_{\X(l)}[\F(l)_0,\G(l)_0,\mathcal{H}(l)_0,\eta(l)]$. To show naturality -- with respect to $\mathcal{H}$ and $\F \Rightarrow \G$ -- we need to establish the equality
\begin{align*}
& \big(\lambda_{\X(l)}[\F(l)_0,\G(l)_0,\mathcal{H}(l)_0,\eta(l)] \circ \target_{\M(l)^n}\big) \circ_{\X(l)} \mathcal{H}(l)_1 = \\
& \;\;\; \big(\F(l)_1^{-1} \Rightarrow_{\X(l)} \G(l)_1\big) \circ_{\X(l)} \big(\lambda_{\X(l)}[\F(l)_0,\G(l)_0,\mathcal{H}(l)_0,\eta(l)] \circ \source_{\M(l)^n}\big)
\end{align*}

The target of the two morphisms is an exponential, so it suffices to check that taking a product of each morphism with the identity and postcomposing with the evaluation morphism yields the same result. Moreover, since $\id_{\X(l)}[\mathcal{H}(l)_0~\circ~\source_{\M(l)^n}]~\times_{\X(l)}~\F(l)_1$ is an isomorphism, it suffices to show that a further precomposition with this isomorphism yields the same result. To this end we observe the chain of equalities below. Equalities $(1)$, $(5)$, $(7)$, $(8)$, and the green part of $(2)$ follow since $\times_{\X(l)}$ suitably commutes with $\circ_{\X(l)}$; equality $(4)$ and the red part of $(2)$ follow by the definition of $\lambda_{\X(l)}$; equality $(3)$ follows by the degeneracy-preservation of $\eta$; and equality $(6)$ follows by the definition of $\Rightarrow_{\X(l)}$.

\resizebox{0.95\linewidth}{!}{
\begin{minipage}{\linewidth}
\begin{align*}
& \eval_{\X(l)}\big[\F(l)_0 \circ \target_{\M(l)^n}, \G(l)_0 \circ\target_{\M(l)^n}\big] \; \circ_{\X(l)} \\ 
& \;\;\; \bigg(\Big(\big(\lambda_{\X(l)}[\F(l)_0,\G(l)_0,\mathcal{H}(l)_0,\eta(l)] \circ \target_{\M(l)^n}\big) \circ_{\X(l)} \mathcal{H}(l)_1\Big) \; \times_{\X(l)} \\
& \;\;\; \id_{\X(l)}[\F(l)_0 \circ \target_{\M(l)^n}]\bigg) \circ_{\X(l)} \big(\id_{\X(l)}[\mathcal{H}(l)_0 \circ \source_{\M(l)^n}] \times_{\X(l)} \F(l)_1\big) \\
\stackrel{(1)}{=} \; & {\color{red} \eval_{\X(l)}\big[\F(l)_0 \circ \target_{\M(l)^n}, \G(l)_0 \circ \target_{\M(l)^n}\big] \; \circ_{\X(l)}} \\
& \;\;\; {\color{red} \Big(\big(\lambda_{\X(l)}[\F(l)_0,\G(l)_0,\mathcal{H}(l)_0,\eta(l)] \circ \target_{\M(l)^n}\big) \times_{\X(l)} \id_{\X(l)}[\F(l)_0 \circ \target_{\M(l)^n}]\Big)} \; \circ_{\X(l)} \\
& \;\;\; {\color{teal} \big(\mathcal{H}(l)_1 \times_{\X(l)} \id_{\X(l)}[\F(l)_0 \circ \target_{\M(l)^n}]\big) \circ_{\X(l)} \big(\id_{\X(l)}[\mathcal{H}(l)_0 \circ \source_{\M(l)^n}] \times_{\X(l)} \F(l)_1\big)} \\
\stackrel{(2)}{=} \; & {\color{red} (\eta(l) \circ \target_{\M(l)^n})} \circ_{\X(l)} {\color{teal}(\mathcal{H}(l)_1 \times_{\X(l)} \F(l)_1)} \\
\stackrel{(3)}{=} \; & \G(l)_1 \circ_{\X(l)} (\eta(l) \circ \source_{\M(l)^n}) \\
\stackrel{(4)}{=} \; & \G(l)_1 \circ_{\X(l)} \eval_{\X(l)} \big[\F(l)_0 \circ \source_{\M(l)^n}, \G(l)_0 \circ \source_{\M(l)^n}\big] \; \circ_{\X(l)} \\
& \;\;\; {\color{blue}\Big(\big(\lambda_{\X(l)}[\F(l)_0,\G(l)_0,\mathcal{H}(l)_0,\eta(l)] \circ \source_{\M(l)^n}\big) \times_{\X(l)} \id_{\X(l)}[\F(l)_0 \circ \source_{\M(l)^n}]\Big)} \\
\stackrel{(5)}{=} \; & \G(l)_1 \circ_{\X(l)} \eval_{\X(l)}\big[\F(l)_0 \circ \source_{\M(l)^n}, \G(l)_0 \circ \source_{\M(l)^n}\big] \; \circ_{\X(l)} \\
& \;\;\; {\color{blue}\Big(\id_{\X(l)}\big[(\F(l)_0 \circ \source_{\M(l)^n}) \Rightarrow_{\X(l)} (\G(l)_0 \circ \source_{\M(l)^n})\big] \times_{\X(l)} \F(l)_1^{-1} \Big)} \; \circ_{\X(l)} \\
& \;\;\; {\color{blue}\Big(\big(\lambda_{\X(l)}[\F(l)_0,\G(l)_0,\mathcal{H}(l)_0,\eta(l)] \circ \source_{\M(l)^n}\big) \times_{\X(l)} \F(l)_1\Big)} \\ 
\stackrel{(6)}{=} \; &\eval_{\X(l)}\big[\F(l)_0 \circ \target_{\M(l)^n}, \G(l)_0 \circ\target_{\M(l)^n}\big] \; \circ_{\X(l)} \\
& \;\;\; \Big(\big(\F(l)_1^{-1} \Rightarrow_{\X(l)} \G(l)_1\big) \times_{\X(l)} \id_{\X(l)}[\F(l)_0 \circ \target_{\M(l)^n}]\Big) \; \circ_{\X(l)} \\
& \;\;\; \Big(\big(\lambda_{\X(l)}[\F(l)_0,\G(l)_0,\mathcal{H}(l)_0,\eta(l)] \circ \source_{\M(l)^n}\big) \times_{\X(l)} \F(l)_1\Big) \\
\stackrel{(7)}{=} \; &\eval_{\X(l)}\big[\F(l)_0 \circ \target_{\M(l)^n}, \G(l)_0 \circ\target_{\M(l)^n}\big] \; \circ_{\X(l)} \\
& \;\;\; {\color{purple} \Big(\big(\F(l)_1^{-1} \Rightarrow_{\X(l)} \G(l)_1\big) \times_{\X(l)} \id_{\X(l)}[\F(l)_0 \circ \target_{\M(l)^n}]\Big) \; \circ_{\X(l)}} \\
& \;\;\; {\color{purple} \Big(\big(\lambda_{\X(l)}[\F(l)_0,\G(l)_0,\mathcal{H}(l)_0,\eta(l)] \circ \source_{\M(l)^n}\big) \times_{\X(l)} \id_{\X(l)}[\F(l)_0 \circ \target_{\M(l)^n}]\Big)} \; \circ_{\X(l)} \\ 
& \;\;\; \big(\id_{\X(l)}[\mathcal{H}(l)_0 \circ \source_{\M(l)^n}] \times_{\X(l)} \F(l)_1\big) \\
\stackrel{(8)}{=} \; &\eval_{\X(l)}\big[\F(l)_0 \circ \target_{\M(l)^n}, \G(l)_0 \circ\target_{\M(l)^n}\big] \circ_{\X(l)} {\color{purple}\bigg(\Big(\big(\F(l)_1^{-1} \Rightarrow_{\X(l)} \G(l)_1\big) \; \circ_{\X(l)}} \\
& \;\;\; {\color{purple}\big(\lambda_{\X(l)}[\F(l)_0,\G(l)_0,\mathcal{H}(l)_0,\eta(l)] \circ \source_{\M(l)^n}\big)\Big) \times_{\X(l)} \id_{\X(l)}[\F(l)_0 \circ \target_{\M(l)^n}]\bigg)} \; \circ_{\X(l)} \\
& \;\;\; \big(\id_{\X(l)}[\mathcal{H}(l)_0 \circ \source_{\M(l)^n}] \times_{\X(l)} \F(l)_1\big)
\end{align*}
\end{minipage}}\pagebreak

To prove that our candidate universal morphism is degeneracy-preserving -- with respect to $\varepsilon_\mathcal{H}$ and $\varepsilon_{\F \Rightarrow \G}$ -- we need to establish the following equality:
\begin{align*}
& \big(\lambda_{\X(1)}[\F(1)_0,\G(1)_0,\mathcal{H}(1)_0,\eta(1)] \circ \M(\dm)_0^n\big) \circ_{\X(1)} \varepsilon_\mathcal{H} = \\
& \;\;\; \big(\varepsilon_\F^{-1} \Rightarrow_{\X(1)} \varepsilon_\G\big) \circ_{\X(1)} \eta_\X^\mathsf{\Rightarrow}[\F(0)_0,\G(0)_0]\; \circ_{\X(1)} \\ 
& \;\;\; \big(\X(\dm)_1 \circ \lambda_{\X(0)}[\F(0)_0,\G(0)_0,\mathcal{H}(0)_0,\eta(0)]\big)
\end{align*}

The target of the two morphisms is an exponential, so it suffices to check that taking a product of each morphism with the identity and postcomposing with the evaluation morphism yields the same result. Moreover, since 
\[\big(\id_{\X(1)}[\X(\dm)_0 \circ \mathcal{H}(0)_0]  \times_{\X(1)} \varepsilon_\F\big) \circ_{\X(1)} \eta_\X^\times[\F(0)_0,\mathcal{H}(0)_0] \]
is an isomorphism, it suffices to show that a further precomposition with this isomorphism yields the same result. To this end we observe the two chains of equalities below. Equalities $(1)$, $(7)$, $(9)$, $(10)$, and the green part of $(2)$ follow since $\times_{\X(1)}$ suitably commutes with $\circ_{\X(1)}$; equality $(4)$ and the red part of $(2)$ follow by the definition of $\lambda_{\X(l)}$; equality $(3)$ follows by the degeneracy-preservation of $\eta$; equality $(5)$ follows by the functoriality of $\X(\dm)$; equality $(6)$ follows by the definition of $\eta_\X^\Rightarrow$; the orange part of equality $(8)$ follows by the definition of $\Rightarrow_{\X(1)}$ on morphisms; and the purple part of equality $(8)$ follows by the naturality of $\eta_\X^\times$. The preservation of face maps is shown by the exact same argument. 

This shows that our candidate universal morphism is a proper morphism. Its universality and uniqueness are obvious by the universal property of $\Rightarrow_{\X(l)}$.

\resizebox{0.95\linewidth}{!}{
\begin{minipage}{\linewidth}
\begin{align*}
& \eval_{\X(1)}\big[\F(1)_0 \circ \M(\dm)_0^n, \G(1)_0 \circ \M(\dm)_0^n\big] \; \circ_{\X(1)} \\
& \;\;\; \bigg(\Big(\big(\lambda_{\X(1)}[\F(1)_0,\G(1)_0,\mathcal{H}(1)_0,\eta(1)] \circ \M(\dm)_0^n\big) \circ_{\X(1)} \varepsilon_\mathcal{H}\Big) \\
& \;\;\; \times_{\X(1)} \id_{\X(1)}[\F(1)_0 \circ \M(\dm)_0^n]\bigg) \circ_{\X(1)} \big(\id_{\X(1)}[\X(\dm)_0 \circ \mathcal{H}(0)_0] \times_{\X(1)} \varepsilon_\F\big) \; \circ_{\X(1)} \\
& \;\;\; \eta_\X^\times[\F(0)_0,\mathcal{H}(0)_0] \\
\stackrel{(1)}{=} \; & {\color{red} \eval_{\X(1)}\big[\F(1)_0 \circ \M(\dm)_0^n, \G(1)_0 \circ \M(\dm)_0^n\big] \; \circ_{\X(1)}} \\ 
& \;\;\; {\color{red} \Big(\big(\lambda_{\X(1)}[\F(1)_0,\G(1)_0,\mathcal{H}(1)_0,\eta(1)] \circ \M(\dm)_0^n\big) \times_{\X(1)} \id_{\X(1)}[\F(1)_0 \circ \M(\dm)_0^n]\Big)} \; \circ_{\X(1)} \\
& \;\;\; {\color{teal}\big(\varepsilon_\mathcal{H} \times_{\X(1)} \id_{\X(1)}[\F(1)_0 \circ \M(\dm)_0^n]\big) \circ_{\X(1)} \big(\id_{\X(1)}[\X(\dm)_0 \circ \mathcal{H}(0)_0] \times_{\X(1)} \varepsilon_\F \big)} \; \circ_{\X(1)} \\
& \;\;\; \eta_\X^\times[\F(0)_0,\mathcal{H}(0)_0] \\
\stackrel{(2)}{=} \; & {\color{red} (\eta(1) \circ \M(\dm)_0^n)} \circ_{\X(1)} {\color{teal} (\varepsilon_\mathcal{H} \times_{\X(1)} \varepsilon_\F)} \circ_{\X(1)} \eta_\X^\times[\F(0)_0,\mathcal{H}(0)_0] \\
\stackrel{(3)}{=} \; & \varepsilon_\G \circ_{\X(1)} (\X(\dm)_1 \circ \eta(0)) \\
\stackrel{(4)}{=} \; & \varepsilon_\G \circ_{\X(1)} \bigg(\X(\dm)_1 \circ \Big(\eval_{\X(0)}[\F(0)_0,\G(0)_0] \; \circ_{\X(0)}  \\
& \;\;\; \big(\lambda_{\X(0)}[\F(0)_0,\G(0)_0,\mathcal{H}(0)_0,\eta(0)] \times_{\X(0)} \id_{\X(0)}[\F(0)_0] \big)\Big)\bigg)\\
\stackrel{(5)}{=} \; & \varepsilon_\G \circ_{\X(1)} \big(\X(\dm)_1 \circ \eval_{\X(0)}[\F(0)_0,\G(0)_0]\big) \; \circ_{\X(1)} \\
& \;\;\; \Big(\X(\dm)_1 \circ \big(\lambda_{\X(0)}[\F(0)_0,\G(0)_0,\mathcal{H}(0)_0,\eta(0)] \times_{\X(0)} \id_{\X(0)}[\F(0)_0]\big)\Big)
\end{align*}
\end{minipage}}\bigskip

\resizebox{0.95\linewidth}{!}{
\begin{minipage}{\linewidth}
\begin{align*}
& \varepsilon_\G \circ_{\X(1)} {\color{blue} \big(\X(\dm)_1 \circ \eval_{\X(0)}[\F(0)_0,\G(0)_0]\big)} \; \circ_{\X(1)} \\
& \;\;\; \Big(\X(\dm)_1 \circ \big(\lambda_{\X(0)}[\F(0)_0,\G(0)_0,\mathcal{H}(0)_0,\eta(0)] \times_{\X(0)} \id_{\X(0)}[\F(0)_0]\big)\Big) \\
\stackrel{(6)}{=} \; & \varepsilon_\G \circ_{\X(1)} {\color{blue} \eval_{\X(1)}\big[\X(\dm)_0 \circ \F(0)_0, \X(\dm)_0 \circ \G(0)_0\big] \; \circ_{\X(1)}} \\
& \;\;\; {\color{blue} \big(\eta_\X^\mathsf{\Rightarrow}[\F(0)_0,\G(0)_0] \times_{\X(1)} \id_{\X(1)}[\X(\dm)_0 \circ \F(0)_0]\big) \; \circ_{\X(1)}} \\ 
& \;\;\; {\color{blue} \eta_\X^\times\big[\F(0)_0 \Rightarrow_{\X(1)} \mathcal{G}(0)_0,\F(0)_0\big]} \; \circ_{\X(1)} \\
& \;\;\; \Big(\X(\dm)_1 \circ \big(\lambda_{\X(0)}[\F(0)_0,\G(0)_0,\mathcal{H}(0)_0,\eta(0)] \times_{\X(0)} \id_{\X(0)}[\F(0)_0]\big)\Big) \\
\stackrel{(7)}{=} \; & {\color{orange} \varepsilon_\G \circ_{\X(1)} \eval_{\X(1)}\big[\X(\dm)_0 \circ \F(0)_0, \X(\dm)_0 \circ \G(0)_0\big] \; \circ_{\X(1)}} \\
& \;\;\; {\color{orange}\Big(\id_{\X(1)}\big[(\X(\dm)_0 \circ \F(0)_0) \Rightarrow_{\X(1)} (\X(\dm)_0 \circ \G(0)_0)\big] \times_{\X(1)} \varepsilon_\F^{-1}\Big)} \; \circ_{\X(1)} \\
& \;\;\; \big(\eta_\X^\mathsf{\Rightarrow}[\F(0)_0,\G(0)_0] \times_{\X(1)} \varepsilon_\F\big) \; \circ_{\X(1)} \\
& \;\;\; {\color{purple} \eta_\X^\times\big[\F(0)_0 \Rightarrow_{\X(1)} \mathcal{G}(0)_0,\F(0)_0\big] \; \circ_{\X(1)}} \\
& \;\;\; {\color{purple} \Big(\X(\dm)_1 \circ \big(\lambda_{\X(0)}[\F(0)_0,\G(0)_0,\mathcal{H}(0)_0,\eta(0)] \times_{\X(0)} \id_{\X(0)}[\F(0)_0]\big)\Big)} \\
\stackrel{(8)}{=} \; & {\color{orange}\eval_{\X(1)}\big[\F(1)_0 \circ \M(\dm)_0^n, \G(1)_0 \circ \M(\dm)_0^n\big] \; \circ_{\X(1)}} \\
& \;\;\; {\color{orange} \Big(\big(\varepsilon_\F^{-1} \Rightarrow_{\X(1)} \varepsilon_\G\big) \times_{\X(1)} \id_{\X(1)}[\F(1)_0 \circ \M(\dm)_0^n]\Big)} \; \circ_{\X(1)} \\
& \;\;\; \big(\eta_\X^\mathsf{\Rightarrow}[\F(0)_0,\G(0)_0] \times_{\X(1)} \varepsilon_\F\big)\; \circ_{\X(1)} \\
& \;\;\; {\color{purple}\Big(\big(\X(\dm)_1 \circ \lambda_{\X(0)}[\F(0)_0,\G(0)_0,\mathcal{H}(0)_0,\eta(0)]\big) \times_{\X(1)} \id_{\X(1)}[\X(\dm)_0 \circ \F(0)_0]\Big) \; \circ_{\X(1)}} \\
& {\color{purple}\;\;\; \eta_\X^\times[\F(0)_0,\mathcal{H}(0)_0]} \\
\stackrel{(9)}{=} \; & \eval_{\X(1)}\big[\F(1)_0 \circ \M(\dm)_0^n, \G(1)_0 \circ \M(\dm)_0^n\big] \; \circ_{\X(1)} \\
& \;\;\; {\color{olive} \Big(\big(\varepsilon_\F^{-1} \Rightarrow_{\X(1)} \varepsilon_\G\big) \times_{\X(1)} \id_{\X(1)}[\F(1)_0 \circ \M(\dm)_0^n]\Big) \; \circ_{\X(1)}} \\
& \;\;\; {\color{olive} \Big(\eta_\X^\mathsf{\Rightarrow}[\F(0)_0,\G(0)_0] \times_{\X(1)} \id_{\X(1)}[\F(1)_0 \circ \M(\dm)_0^n]\Big) \; \circ_{\X(1)}} \\
& \;\;\; {\color{olive} \Big(\big(\X(\dm)_1 \circ \lambda_{\X(0)}[\F(0)_0,\G(0)_0,\mathcal{H}(0)_0,\eta(0)]\big) \times_{\X(1)} \id_{\X(1)}[\F(1)_0 \circ \M(\dm)_0^n]\Big)} \; \circ_{\X(1)} \\
& \;\;\; \big(\id_{\X(1)}[\X(\dm)_0 \circ \mathcal{H}(0)_0] \times_{\X(1)} \varepsilon_\F\big) \circ_{\X(1)} \eta_\X^\times[\F(0)_0,\mathcal{H}(0)_0] \\
\stackrel{(10)}{=} \; & \eval_{\X(1)}\big[\F(1)_0 \circ \M(\dm)_0^n, \G(1)_0 \circ \M(\dm)_0^n\big] \; \circ_{\X(1)} \\
& \;\;\;{\color{olive}\bigg(\Big(\big(\varepsilon_\F^{-1} \Rightarrow_{\X(1)} \varepsilon_\G\big) \circ_{\X(1)} \eta_\X^\mathsf{\Rightarrow}[\F(0)_0,\G(0)_0]\; \circ_{\X(1)}} \\
& \;\;\; {\color{olive} \big(\X(\dm)_1 \circ \lambda_{\X(1)}[\F(0)_0,\G(0)_0,\mathcal{H}(0)_0,\eta(0)]\big)\Big) \times_{\X(1)} \id_{\X(1)}[\F(1)_0 \circ \M(\dm)_0^n]\bigg)}  \; \circ_{\X(1)} \\
& \;\;\; \big(\id_{\X(1)}[\X(\dm)_0 \circ \mathcal{H}(0)_0] \times_{\X(1)} \varepsilon_\F\big) \circ_{\X(1)} \eta_\X^\times[\F(0)_0,\mathcal{H}(0)_0]
\end{align*}
\end{minipage}}\pagebreak
\end{proof}

\begin{definition}\label{def:ccrgcwi}
A reflexive graph category with isomorphisms is \emph{cartesian closed} if it has terminal objects, products, and exponentials, all
stable under face maps and degeneracies.
\end{definition}

\begin{example}\label{ex:per-cont}[PER model, continued]
Terminal objects, products, and exponentials are defined for
$\mathcal{R}_\mathit{PER}$ in the obvious ways, inheriting from the
corresponding constructs on PERs. It is not hard to check that all of
these constructs are preserved on the nose by the two face maps
(projections) and the degeneracy (equality functor), and thus, in our
terminology, are stable under face maps and degeneracies.
\end{example}

\begin{example}\label{ex:rey-cont}[Both versions of Reynolds' model, continued]
Here, too, terminal objects, products, and exponentials are defined
for $\mathcal{R}_\mathit{REY}$ and $\mathcal{R}_\mathit{CREY}$ in the
obvious ways, relating two pairs iff their first and second components
are related, and two functions iff they map related arguments to
related results. It is easy to see that all of these
constructs are preserved on the nose ({\em i.e.}, up to \emph{definitional} equality) by the projections, and thus are stable
under face maps. Unlike in the PER model though, they are only
preserved by the equality functor $\Eq$ up to (the canonical)
isomorphism. For example, as discussed just after
Definition~\ref{def:rgfmp}, the two types $\Id((a,b),(c,d))$ and
$\Id(a,c) \times \Id(b,d)$ for $(a,b), (c,d) : A \times B$ are not necessarily identical, although they are isomorphic under the canonical
(iso)morphism $\eta^\times[A,B] : \Eq(A \times B) \to \Eq(A) \times
\Eq(B)$. A similar situation arises for function types $A \to B$: by
function extensionality, $\Id(f,g)$ and $\Pi_{a,a':A} \Id(f(a),
g(a'))$ are isomorphic, but not necessarily identical, via
$\eta^\Rightarrow[A,B]$.  Nevertheless, we still get stability under
degeneracies since we explicitly allowed for this possibility in
Definition~\ref{def:exponentials1}.
\end{example}

\section{Reflexive Graph Models of Parametricity}\label{sec:main}
As Examples~\ref{ex:per-cont} and~\ref{ex:rey-cont} show, cartesian
closed reflexive graph categories with isomorphisms suitably
generalize the structure of sets and relations. Moreover, they allow
us to interpret unit, product, and function types in a natural
way. To show this, we introduce the following terminology, presented in a form more general than we need for interpreting the simply-typed fragment of System F but paralleling the later terminology used for interpreting the impredicative fragment. 

\begin{definition}
\emph{A $\lambda^\to$-fibration} is a split fibration $U : \mathcal{E} \to \mathcal{B}$ satisfying the following
properties:
\begin{enumerate}
\item The objects of $\mathcal{B}$ are in bijection with $\nat$, with $0$ serving as a terminal object in $\mathcal{B}$, $1$ serving as a split generic object for $U$, and $n + 1$ serving as a product of $n$ and $1$.
\item Every fiber $\E_n$ for $n$ in $\mathcal{B}$ is cartesian closed, with a terminal object $\mathsf{1}_n$, products $\times_n$, and exponentials $\Rightarrow_n$.
\item Beck-Chevalley: for any $f : n \to m$ in $\mathcal{B}$ and objects $X,Y$ in $\E_m$, the canonical morphisms below are isomorphisms: 
\begin{align*}
& \theta_\mathsf{1}(f) : f^*(\mathsf{1}_m) \to \mathsf{1}_n \\ 
& \theta_\times(f,X,Y) : f^*(X \times_m Y) \to \big(f^*(X) \times_n f^*(Y)\big)\\ 
& \theta_\Rightarrow(f,X,Y) : f^*(X \Rightarrow_m Y) \to \big(f^*(X) \Rightarrow_n f^*(Y)\big)
\end{align*}
\end{enumerate}
A $\lambda^\to$-fibration is \emph{split} if these canonical morphisms are identities.
\end{definition}

\noindent Using a similar idea as in the proof of Lemma 5.2.4 of \cite{jac99}, we can show:
\begin{lemma}\label{lem:strict}
Every $\lambda^\to$-fibration is equivalent to a split $\lambda^\to$-fibration in a canonical way.
\end{lemma}

\noindent We now come to our main technical lemma:

\begin{theorem}\label{lem:cartesian_split}
Given a cartesian closed reflexive graph category $\R$ with isomorphisms, the forgetful functor from the category $\int_n \,\M^n \to \M$ to $\ctx(\R)$ is a split $\lambda^\to$-fibration.
\end{theorem}

To interpret $\forall$-types we need to know that, in the forgetful
fibration from Lemma~\ref{lem:cartesian_split}, each weakening functor
induced by the first projection from $n + 1$ to $n$ for $n \in \nat$
has a right adjoint $\forall_n$. Here we differ from \cite{dr04},
where only $\forall_0$ is required, with the intention that
$\forall_n$ can be derived from $\forall_0$ using partial
application. We observe that this approach does not appear to work
since a partial application of an indexed functor is not necessarily
an indexed functor. Hence we require an entire family of adjoints
$\forall_n$.

\begin{example}[PER model, continued]\label{ex:per2}
Define the adjoint $\forall_n$ by
\begin{align*}
&\forall_n\,\F(0) \; \overline{A} \coloneqq \big\{ (m,k) \; | \; \text{ for all} \; A, (m,k) \in \F(0) \, (\overline{A},A), \\ 
& \;\;\; \text{and for all} \; R, \langle m,k \rangle \in \F(1)(\overline{\Eq\,A},R) \big\} \\
&\forall_n\,\F(1) \; \overline{R} \coloneqq \Big(\big(\forall_n \, \F(0) \; \overline{R_\mathsf{d}}, \forall_n \, \F(0) \; \overline{R_\mathsf{c}}\big), \\  
& \;\;\; \big\{ m ~|~\text{ for all} \; R, m \in \F(1) \,(\overline{R},R) \big\} \Big)
\end{align*}
where for any relation $R \coloneqq ((A_\mathsf{d},A_\mathsf{c}),R_A)$
we write $R_\mathsf{d}$ for $A_\mathsf{d}$ and $R_\mathsf{c}$ for
$A_\mathsf{c}$. We will employ a similar convention for Reynolds'
model. To define $\forall_n$ on a morphism $\eta : \F \to \Gcal$, we
put 
\begin{align*}
& \forall_n\,\eta(0) \; \overline{A} \coloneqq \Big(\big(\forall_n\,\F(0) \; \overline{A}, \forall_n\,\G(0) \; \overline{A}\big), \{ m \cdot 0 \}_{(\forall_n\,\F(0) \; \overline{A}) \to (\forall_n\,\G(0) \; \overline{A}})\Big)
\end{align*}

Here $m$ is any natural number realizing $\eta(0)~\overline{A}$. Crucial observations are
that all natural transformations are ``uniformly realized'' in the
sense that there is a natural number realizing each such
transformation, and since all PERs are defined to be realized by all
natural numbers, each is suitably uniform. In particular, if $\eta$
were not uniformly realized in the above sense then $\forall_n$ would
not be well-defined on morphisms. These observations can be used to
show that, in the category-theoretic setting (rather than the setting
of $\omega$-sets), the family of adjoints $\forall$ cannot exist
precisely because {\em ad hoc} natural transformations --- {\em i.e.},
natural transformations that are not uniformly realizable, even though
each of their components may indeed be realizable --- are not
excluded.
\end{example}	

\begin{example}[Reynolds' model, continued]\label{ex:rey2}
On the set level, the adjoint $\forall_n$ is defined as follows:
\begin{align*}
\forall_n \, \F(0) \, \overline{A} \coloneqq \Big\{ & f_0 : \Pi_{A:\U} \F(0)(\overline{A},A) \; \& \\
& f_1 : \Pi_{R : \mathsf{R}_0} \F(1)\,(\overline{\Eq \,A},R) \, (f_0(R_\mathsf{d}), f_0(R_\mathsf{d}))\Big\}
\end{align*}
On the relation level, we define $\forall_n \F(1) \, \overline{R}$ to
be the relation with domain $\forall_n\F(0) \, \overline{R_0}$ and
codomain $\forall_n\F(0) \, \overline{R_1}$ mapping
$((f_0,f_1),(g_0,g_1))$ to 
\[\Pi_{R : \mathsf{R}_0} \F(1)\,(\overline{R},R) \, (f_0(R_\mathsf{d}),g_0(R_\mathsf{c}))\]

To see that the above definition indeed gives a degeneracy-preserving
reflexive graph functor, fix $\overline{A}$. We want to show that the
two relations $\Eq \, (\forall_n \, \F(0) \, \overline{A})$ and
$\forall_n \, \F(1) \, \overline{\Eq(A)}$ are isomorphic. The domains
and codomains of these relations are all the same --- $\forall_n \,
\F(0) \, \overline{A}$ --- so we let both of the underlying maps of the
isomorphism be identities (as also required by the coherence condition
on the isomorphism and, independently, the definition of a relevant
isomorphism). Fix $((f_0,f_1),(g_0,g_1)) : (\forall_n \, \F(0) \,
\overline{A}) \times (\forall_n \, \F(0) \, \overline{A})$. We need
functions going back and forth between the types $\Id((f_0,f_1),(g_0,g_1))$ and
$\Pi_{R : \mathsf{R}_0} \F(1)\,(\overline{\Eq(A)},R) \,
(f_0(R_0),g_0(R_1))$. Such functions will automatically be mutually
inverse since the types in question are propositions.

Going from left to right is easy using $\Id$-induction and $f_1$. To
go from right to left, fix $\phi : \Pi_{R : \mathsf{R}_0}
\F(1)\,(\overline{\Eq(A)},R) \, (f_0(R_0),g_0(R_1))$. To show
$\Id((f_0,f_1),(g_0,g_1))$ it suffices to show $\Id(f_0,g_0)$ since
the type of $g_1$ (or $f_1$) is a proposition. By function
extensionality, it suffices to show pointwise equality between $f_0$
and $g_0$. So fix $B$. The only thing we can do with $\phi$ is to
apply it to $\Eq(B)$, which gives us $\phi(\Eq(B)) :
\F(1)\,(\overline{\Eq(A)},\Eq(B)) \, (f_0(B),g_0(B))$. The relation
$\F(1)\,(\overline{\Eq(A)},\Eq(B))$ is isomorphic to $\Eq\,
\F(0)\,(\overline{A},B)$ via
$\varepsilon_\F(\overline{A},B)^{-1}$. Applying
$\varepsilon_\F(\overline{A},B)^{-1}$ to $(f_0(B),g_0(B))$ and
$\phi(\Eq(B))$ thus gives us
$\Id\big(\varepsilon_\F(\overline{A},B)^{-1}_\top \, f_0(A),
\varepsilon_\F(\overline{A},B)^{-1}_\bot \, g_0(B)\big)$. The
coherence condition on $\varepsilon_\F$ tells us that the respective
images $\varepsilon_\F(\overline{A},B)_\top$ and
$\varepsilon_\F(\overline{A},B)_\bot$ of
$\varepsilon_\F(\overline{A},B)$ under the two face maps are the
identity on $\F(0)(\overline{A},B)$, and thus are
$\varepsilon_\F(\overline{A},B)^{-1}_\top$ and
$\varepsilon_\F(\overline{A},B)^{-1}_\bot$. This gives
$\Id(f_0(A),g_0(B))$ as desired.
\end{example}

\begin{example}[A categorical version of Reynolds' model, continued]\label{ex:crey2}
On the set level, the adjoint $\forall_n$ is defined as follows: 

\begin{align*}
\forall_n \, \F(0) \, \overline{A} \coloneqq \Big\{ & f_0 : \Pi_{A:\U} \F(0)(\overline{A},A) \; \& \\
& f_1 : \Pi_{R : \mathsf{R}_0} \F(1)\,(\overline{\Eq \,A},R) \, (f_0(R_\mathsf{d}), f_0(R_\mathsf{c})) \; \& \\
& \Pi_{i : \M(0)_1} \F(0)\big(\overline{\mathsf{id}_{\M(0)}(A)},i\big) \, f_0(i_\mathsf{d}) = f_0(i_\mathsf{c}) \Big\}
\end{align*}
The last condition says that $f_0$ is functorial in its argument, in the
sense that if $i$ is an isomorphism between two types $A,B : \Set_0$,
then $f_0(A)$ and $f_0(B)$ are suitably related via the obvious
isomorphism between $\F(0)\,(\overline{A},A)$ and
$\F(0)\,(\overline{A},B)$. This condition, which does not have an
analogue in the set-theoretic presentation of Reynolds' model, is
needed because we do not work with discrete domains (\emph{e.g.}, we
use $\F : \M^n \to \M$ rather than $\F : |\M|^n \to \M$), as is common
in other presentations of parametricity. A very similar condition does
appear, {\em e.g.}, in the definition of parametric limits for the
category of sets in \cite{dr04}. The analogous condition asserting the
functoriality of $f_1$ is automatically satisfied since the codomain
of $f_1$ is a proposition. On relations, we use the same
definition as in Example~\ref{ex:rey2}.
\end{example}

\begin{definition}
\emph{A $\lambda 2$-fibration} is a $\lambda^\to$-fibration $U : \mathcal{E} \to \mathcal{B}$ satisfying the following
properties:
\begin{enumerate}
\item For each $n$ in $\mathcal{B}$, the weakening functor induced by the first projection from $n + 1$ to $1$ has a right adjoint $\forall_n$.
\item Beck-Chevalley: for any $f : n \to m$ in $\mathcal{B}$ and object $X$ in $\E_m$, the canonical morphism below is an isomorphism:
\[\theta_\forall(f,X) : f^*(\forall_m(X)) \to \forall_n((f \times \id)^*(X))\]
\end{enumerate}
A $\lambda 2$-fibration is \emph{split} if it is a split $\lambda^\to$-fibration and the canonical morphism above is the identity.
\end{definition}

\noindent Seely \cite{see87} essentially showed the following:
\begin{theorem}[Seely]\label{prop:model}
Every split $\lambda 2$-fibration $U : \mathcal{E} \to \mathcal{B}$ gives a sound model of System F in which:
\begin{itemize}
\item every type context $\Gamma$ is interpreted as an object
  $\sem{\Gamma}$ in $\mathcal{B}$
\item every type $\Gamma \vdash T$ is interpreted as an object
  $\sem{\Gamma \vdash T}$ in the fiber over $\sem{\Gamma}$
\item every term context $\Gamma;\Delta$ is interpreted as an object
  $\sem{\Gamma \vdash \Delta}$ in the fiber over $\sem{\Gamma}$
\item every term $\Gamma;\Delta \vdash t : T$ is interpreted as a
  morphism $\sem{\Gamma;\Delta \vdash t : T}$ from
  $\sem{\Gamma;\Delta}$ to $\sem{\Gamma \vdash T}$ in the fiber over $\sem{\Gamma}$
\end{itemize}
\end{theorem}

\noindent A (not necessarily split) $\lambda 2$-fibration also gives a sound model of System F, due to the following:
\begin{lemma}
Every $\lambda 2$-fibration is equivalent to a split $\lambda 2$-fibration in a canonical way.
\end{lemma}

We now want to specify when a model of System F given by a $\lambda 2$-fibration is relationally parametric. If $\mathcal{R}$ is a cartesian closed reflexive graph category with isomorphisms, we denote by $\mathsf{F}(\mathcal{R})$ the $\lambda^\to$-fibration induced by $\R$ as in Theorem~\ref{lem:cartesian_split}. To formulate an abstract definition of a parametric model, we will appropriately relate a $\lambda 2$-fibration $U$ to $\mathsf{F}(\mathcal{R})$. To see how, we revisit the simplest model, namely the System F term model. In the $\lambda 2$-fibration $U_\mathit{term}$ corresponding to the term model, the fiber over $n \in \nat$ consists of types and terms with $n$ free type variables. Let $\mathcal{U}$ be the category consisting of closed System F types and terms between them. Then $\mathcal{U}$ induces a $\lambda^\to$-fibration, $U_\mathit{set}$, whose fiber over $n$ consists of functors $|\mathcal{U}|^n \to \mathcal{U}$ and natural transformations between them. 

A type $\overline{\alpha} \vdash T$ with $n$ free variables can now be seen as functor $|\mathcal{U}|^n \to \mathcal{U}$, and a term $\overline{\alpha}; x : S \vdash t : T$ as a natural transformation
between $S$ and $T$. We thus have a morphism of $\lambda^\to$-fibrations $\mu : U_\textit{term} \to U_{\textit{set}}$. However, unlike $U_{\textit{term}}$, $U_\textit{set}$ does not admit the family
of adjoints required to make it a $\lambda 2$-fibration. Still, we can view $U_\mathit{term}$ as a version of $U_\mathit{set}$ that ``enriches'' the functors and natural transformations with enough extra information to ensure that the desired adjoints exist: in this example, the information that the maps involved are not {\em ad hoc}, but come from syntax. Since these adjunctions are only applicable to non-empty contexts, no such ``enrichment'' should be necessary for objects and morphisms over the \emph{terminal} object. And indeed, the restriction
of $\mu$ to the fibers over the respective terminal objects is clearly
an equivalence. These observations echo those immediately following
Definition~\ref{def:rg}, and motivate our main definition:

\begin{definition}\label{def:parametric-model}
Let $\mathcal{R}$ be a cartesian closed reflexive graph category with
isomorphisms. A {\em parametric model of System F over $\mathcal{R}$}
is a $\lambda 2$-fibration $U$ together with a morphism $\mu : U \to \mathsf{F}(\mathcal{R})$ of $\lambda^\to$-fibrations whose restriction to the fibers of $U$ and $\mathsf{F}(\mathcal{R})$ over the 
terminal objects is full, faithful, and essentially surjective.
\end{definition}

Our main theorem shows that the definition of a parametric model is indeed sensible:
\begin{theorem}
Every parametric model of System F over a cartesian closed reflexive graph category $(\X,(\M,\I))$ with isomorphisms, as specified in Definition~\ref{def:parametric-model}, is a sound model in which:
\begin{itemize}
\item every type $\Gamma \vdash T$ can be seen as a face map- and degeneracy-preserving reflexive graph functor $\sem{\Gamma \vdash T} : \M^{|\Gamma|} \to \M$
\item every term $\Gamma; \Delta \vdash t : T$ can be seen as a face map- and degeneracy-preserving reflexive graph natural transformation $\sem{\Gamma; \Delta \vdash t:T} : \sem{\Gamma \vdash
    \Delta} \to \sem{\Gamma \vdash T}$, with the domain and codomain seen as reflexive graph functors into $\X$
\end{itemize}
\end{theorem}

\begin{theorem}[PER model]
Let $\mathcal{R}_\mathit{PER}$ be the cartesian closed reflexive
graph category with isomorphisms defined in
Examples~\ref{ex:per},~\ref{ex:per1}, and~\ref{ex:per-cont}. The
family of adjoints defined in Example~\ref{ex:per2} makes
$\mathsf{F}(\mathcal{R}_\mathit{PER})$ into a $\lambda
2$-fibration, and hence into a parametric model of
System F over $\mathcal{R}_\mathit{PER}$.
\end{theorem}

\begin{theorem}[Reynolds' model]
Let $\mathcal{R}_\mathit{REY}$ be the reflexive graph category with
isomorphisms defined in Examples~\ref{ex:reynolds},~\ref{ex:rey1},
and~\ref{ex:rey-cont}. The family of adjoints defined in
Example~\ref{ex:rey2} makes $\mathsf{F}(\mathcal{R}_\mathit{REY})$
into a $\lambda 2$-fibration, and hence into a parametric model of System F over $\mathcal{R}_\mathit{REY}$.
\end{theorem}

\begin{theorem}[A categorical version of Reynolds' model]
Let $\mathcal{R}_\mathit{CREY}$ be the reflexive graph category with
isomorphisms defined in Examples~\ref{ex:reynolds},~\ref{ex:crey1}, and~\ref{ex:rey-cont}. The family of adjoints defined in Example~\ref{ex:crey2} makes $\mathsf{F}(\mathcal{R}_\mathit{CREY})$
into a $\lambda 2$-fibration, and hence into a parametric model of System F over $\mathcal{R}_\mathit{CREY}$.
\end{theorem}


\section{A Proof-Relevant Model of Parametricity}\label{sec:proof-relevant}
We now describe a proof-relevant version of Reynolds' model, in which
witnesses of relatedness are themselves related. The construction
of such a model is the subject of \cite{orsanigo_thesis}, but the
development there seems to contain a major technical gap. Specifically, it is unclear how to prove the $\forall$-case in Lemma 9.4 in \cite{orsanigo_thesis}, due to the fact that when types are interpreted as discrete functors $|\X|^n \to \X$, the reindexing of a degeneracy-preserving functor might not be degeneracy-preserving. We already observed this in the introduction but this issue is not addressed in \cite{orsanigo_thesis} and the proof of the lemma is not given there. Since this lemma is crucial to the soundness of the interpretation, it is unknown whether the result of \cite{orsanigo_thesis} can be salvaged as-is. For this reason, we only reuse the main ideas of \cite{orsanigo_thesis} for handling the higher dimensional structure and otherwise proceed independently.

\begin{example}\label{ex:pf-rel}
We use the same ambient category as in Example~\ref{ex:reynolds} and
reuse the (internal) category $\Set$ of types. The category
$\mathsf{R}$ of relations is almost the same as in
Example~\ref{ex:reynolds}, except that relations are now
proof-relevant, \emph{i.e.}, $\mathsf{R}_0 \coloneqq \Sigma_{A,B:\Set}
A \times B \to \Ucal$. As before, we have two face maps $\fm_\top,\fm_\bot : \rel \to \set$ projecting out the domain and codomain of a relation and a degeneracy $\Eq : \set \to \rel$ constructing the equality relation. Given relations $R$ on $A$ and $B$ and $S$ on $C$ and $D$, to relate two witnesses $p : R(a,b)$ and $q : S(c,d)$ we should know {\em a priori} how $a$ relates to $c$ and $b$ to $d$. This
motivates defining the category $\mathsf{2R}$ of \emph{2-relations}, whose objects $Q$ are tuples $(Q^{0\top}, Q^{1\top}, Q^{0\bot}, Q^{1\bot})$ of relations forming a square

\begin{center}
\scalebox{0.9}{
\begin{tikzpicture}
\node (N0) at (0,1.5) {$A$};
\node (N1) at (0,0) {$C$};
\node (N2) at (2,1.5) {$B$};
\node (N3) at (2,0) {$D$};
\draw[->] (N0) -- node[left]{$Q^{1\top}$} (N1);
\draw[->] (N0) -- node[above]{$Q^{0\top}$} (N2);
\draw[->] (N1) -- node[below]{$Q^{0\bot}$} (N3);
\draw[->] (N2) -- node[right]{$Q^{1\bot}$} (N3);
\end{tikzpicture}}
\end{center}
together with a $\Prop$-valued predicate (also denoted $Q$) on the type of tuples of the form $((a,b,c,d),(p,q,r,s))$, where $p : Q^{0\top}(a,b)$, $q : Q^{1\top}(a,c)$, $r : Q^{0\bot}(c,d)$, and $s : Q^{1\bot}(b,d)$. This gives four face maps $\fm_{0\top},\fm_{0\bot},\fm_{1\top},\fm_{1\bot}: \mathsf{2R} \to \mathsf{R}$, one for each edge. We have two degeneracies from $\mathsf{R}$ to $\mathsf{2R}$, one replicating a relation $R$ horizontally and one vertically. More precisely, given $R$, we obtain the 2-relation $\mathsf{Eq}_=(R)$ by placing $R$ on top and bottom, with equality relations $\Eq(R_0)$ and $\Eq(R_1)$ as vertical edges, and mapping $\big((a,b,a,b),(p,-,r,-)\big)$ to $\Id(p,r)$. The symmetric version $\mathsf{Eq}_\parallel(R)$ places $R$ on left and right and assumes equality relations as horizontal edges. But we also have two other ways of turning a relation $R$ into a 2-relation: the functor $\cm_\top$ places $R$ on top and left, and $\cm_\bot(R)$ places $R$ on bottom and right, filling the remaining edges with equalities. The functors
$\cm_\top$ and $\cm_\bot$ are called \emph{connections}. We define terminal objects, products, exponentials, and isomorphisms in the obvious way.

Just like in Reynolds' model, we have $\fm_\star~\circ~\Eq = \id$. We also have further equalities:
\begin{itemize}
\item $\fm_{0 \star} \circ \Eq_= = \id$
\item $\fm_{1\star} \circ \Eq_= = \Eq \circ \fm_\star$ \; for a fixed $\star \in \two$
\item $\fm_{1 \star} \circ \Eq_\parallel = \id$
\item $\fm_{0 \star} \circ \Eq_\parallel = \Eq \circ \fm_\star$ \; for a fixed $\star \in \two$
\item $\fm_{l \star} \circ \cm_\star = \id$ \; for $l \in \{0,1\}$ and a fixed $\star \in \two$
\item $\fm_{l \overline{\star}} \circ \cm_\star = \Eq \circ \fm_{\overline{\star}}$ \; for $l \in \{0,1\}$ and a fixed $\star \in \two$
\end{itemize} 
Moreover, the compositions $\Eq_= \circ \Eq$, $\Eq_\parallel \circ \Eq$, $\cm_\top \circ \Eq$, $\cm_\bot \circ \Eq$ are all naturally isomorphic.
 
The structure described above induces two $\lambda^\to$-fibrations of interest: the first one is induced by combining the first two levels, the categories $\set$ and $\mathsf{R}$, into a cartesian closed
reflexive graph category with isomorphisms $\mathsf{R}_\mathit{PREY}$; this is the fibration $\mathsf{F}(\mathsf{R}_\mathit{PREY})$. We recall that the objects of $\mathsf{F}(\mathsf{R}_\mathit{PREY})$
over $n$ are pairs $\{\F(l) : \M(l)^n \to \M(l) \}_{l \in \{0,1\}}$ of functors that commute with the two face maps from $\mathsf{R}$ to $\set$ on the nose, as well as with the degeneracy $\Eq$ up to a
suitably coherent natural isomorphism $\varepsilon_\F$. The morphisms are pairs $\{\eta(l) : \F(l) \to \Gcal(l)\}_{l \in \{0,1\}}$ of natural transformations that respect both face maps from $\mathsf{R}$ to
$\set$ and the degeneracy $\Eq$. 

The second fibration, which we call $\mathsf{F}_{\mathit{2D}}$, is induced in much the same way, but taking into account all three levels. This means that the objects over $n$ are triples $\{\F(l) : \M(l)^n \to \M(l) \}_{l \in \{0,1,2\}}$ of functors that commute with \emph{all} face maps -- the two from from $\mathsf{R}$ to $\set$ as well as the four from $\mathsf{2R}$ to $\mathsf{R}$ -- on the nose and \emph{all} degeneracies $\Eq,\Eq_=,\Eq_\parallel$ and connections $\cm_\top$, $\cm_\bot$ up to suitably coherent natural isomorphisms. Here ``suitably coherent'' means taking into account not only the equality $\fm_\star \circ \Eq = \id$ but the additional equalities involving $\Eq_=,\Eq_\parallel,\cm_\top,\cm_\bot$ as well. For example, the image of the isomorphism witnessing the commutativity of $\F$ with $\Eq_=$ under the face map $\fm_{1\star}$ must be precisely $\varepsilon_\F \circ \fm_\star$. Analogously, the morphisms are triples $\{\eta(l) : \F(l) \to \Gcal(l)\}_{l \in \{0,1,2\}}$ of natural transformations that respect all face maps, degeneracies, and connections. We have the obvious forgetful morphism of $\lambda^\to$-fibrations from $\mathsf{F}_{\mathit{2D}}$ to $\mathsf{F}(\mathsf{R}_\mathit{PREY})$ that only retains the structure pertaining to levels 0 and 1.

The fibration $\mathsf{F}_{\mathit{2D}}$ admits a family
of adjoints to weakening functors as follows. The adjoint $\forall_n
\,\F(0) \; \overline{A}$ is the type
\begin{align*}
\big\{ & f_0 : \Pi_{A:\Ucal} \F(0)(\overline{A},A) \, \& \\
& f_1 : \Pi_{R : \mathsf{R}_0} \F(1)(\overline{\Eq \,A},R) \, (f_0(R_\mathsf{d}), f_0(R_\mathsf{d})) \, \& \\
& f_2 : \Pi_{Q : \mathsf{2R}_0} \F(2)(\overline{\Eq_=(\Eq(A))} ,Q) \, 
\big(\big(f_0 \,  Q^{0\top}_\mathsf{d}, f_0 \,  Q^{0\top}_\mathsf{c}, f_0 \,  Q^{1\top}_\mathsf{c}, f_0 \,  Q^{0 \bot}_\mathsf{c}), \\
& \;\;\;\;\;\;\;\;\;\;\;\;\;\;\;\;\;\;\;\;\;\;\;\;\;\;\;\;\;\;\;\;\;\;\;\;\;\;\;\;\;\;\;\;\;\;\;\;\;\;\;\;\;\big(f_1 \,  Q^{0\top},  f_1 \,  Q^{1\top}, f_1 \,  Q^{0\bot}, f_1 \, Q^{1\bot}\big)\big) \, \& \\
& \Pi_{i : M(0)_1} \F(0)(\overline{\mathsf{id}_{\M(0)}(A)},i) \, 
f_0(i_\mathsf{c}) = f_0(i_\mathsf{d}) \,  \& \\
& \Pi_{i : M(1)_1} \F(1) \big(\overline{\mathsf{id}_{\M(1)}(\mathsf{Eq} \,
A)},i\big) \, \big(f_0 \,  (i_\mathsf{d})_\mathsf{d}, f_0 \,  (i_\mathsf{d})_\mathsf{d}\big) \, f_1(i_\mathsf{d}) = f_1(i_\mathsf{c}) \big\}
\end{align*}

In the type of $f_2$, we could have just as well used any of the other functors $\Eq_\parallel$, $\cm_\top$, $\cm_\bot$ instead of $\Eq_=$ since as observed above, their compositions with $\Eq$ are all naturally isomorphic. We next define $\forall_n \, \F(1) \, \overline{R}$ to be the relation with domain $\forall_n \, \F(0) \, \overline{R_\mathsf{d}}$ and codomain $\forall_n \, \F(0) \, \overline{R_\mathsf{c}}$ mapping $((f_0,f_1,f_2),(g_0,g_1,g_2))$ to
\begin{align*}
\big\{ & \phi_{\color{white} \Eq_=} : \Pi_{R:\mathsf{R}_0} \F(1)\, (\overline{R},R) \, (f_0(R_\mathsf{d}),g_0(R_\mathsf{c})) \,  \& \\
& \phi_{\Eq_=} : \Pi_{Q : \mathsf{2R}_0} \F(2)\, (\overline{\Eq_=\,R} ,Q) \,
\big(\big(f_0 \,  Q^{0\top}_\mathsf{d}, f_0 \,  Q^{0\top}_\mathsf{c}, g_0 \,  Q^{1\top}_\mathsf{c}, g_0 \, Q^{0\bot}_\mathsf{c}), \\
& \;\;\;\;\;\;\;\;\;\;\;\;\;\;\;\;\;\;\;\;\;\;\;\;\;\;\;\;\;\;\;\;\;\;\;\;\;\;\;\;\;\;\;\;\;\;\;\;~\big(f_1 \,  Q^{0\top},  \phi \,  Q^{1\top}, g_1 \,  Q^{0\bot}, \phi \,  Q^{1\bot}\big)\big) \, \& \\
& \phi_{\Eq_\parallel} : \Pi_{Q : \mathsf{2R}_0} \F(2)\, (\overline{\Eq_\parallel\,R} ,Q) \, \big(\big(f_0 \,  Q^{0\top}_\mathsf{d}, g_0 \,  Q^{0\top}_\mathsf{c}, f_0 \,  Q^{1\top}_\mathsf{c}, g_0 \, Q^{0\bot}_\mathsf{c}), \\
& \;\;\;\;\;\;\;\;\;\;\;\;\;\;\;\;\;\;\;\;\;\;\;\;\;\;\;\;\;\;\;\;\;\;\;\;\;\;\;\;\;\;\;\;\;\;~\big(\phi \,  Q^{0\top},  f_1 \,  Q^{1\top}, \phi \,  Q^{0\bot}, g_1 \,  Q^{1\bot}\big)\big) \, \& \\
& \phi_{\cm_\top} : \Pi_{Q : \mathsf{2R}_0} \F(2)\, (\overline{\cm_\top\,R} ,Q) \,
\big(\big(f_0 \,  Q^{0\top}_\mathsf{d}, g_0 \,  Q^{0\top}_\mathsf{c}, g_0 \,  Q^{1\top}_\mathsf{c}, g_0 \, Q^{0\bot}_\mathsf{c}), \\
& \;\;\;\;\;\;\;\;\;\;\;\;\;\;\;\;\;\;\;\;\;\;\;\;\;\;\;\;\;\;\;\;\;\;\;\;\;\;\;\;\;\;\;\;\;\;~\big(\phi \,  Q^{0\top},  \phi \,  Q^{1\top}, g_1\,  Q^{0\bot}, g_1 \,  Q^{1\bot}\big)\big) \, \& \\
& \phi_{\cm_\bot} : \Pi_{Q : \mathsf{2R}_0} \F(2)\, (\overline{\cm_\bot\,R} ,Q) \,
\big(\big(f_0 \,  Q^{0\top}_\mathsf{d}, f_0 \,  Q^{0\top}_\mathsf{c}, f_0 \,  Q^{1\top}_\mathsf{c}, g_0 \, Q^{0\bot}_\mathsf{c}), \\
& \;\;\;\;\;\;\;\;\;\;\;\;\;\;\;\;\;\;\;\;\;\;\;\;\;\;\;\;\;\;\;\;\;\;\;\;\;\;\;\;\;\;\;\;\;\;~\big(f_1 \,  Q^{0\top},  f_1 \,  Q^{1\top}, \phi\,  Q^{0\bot}, \phi \,  Q^{1\bot}\big)\big) \, \& \\
& \Pi_{i : \M(1)_1} \F(1) (\overline{\mathsf{id}_{\M(1)}(R)},i) \,  \big(f_0 \, (i_\mathsf{d})_\mathsf{d}, g_0 \,  (i_\mathsf{d})_\mathsf{c}\big) \, \phi_1(i_\mathsf{d}) = \phi_1(i_\mathsf{c}) \,  \big\}
\end{align*}
The component $\phi_{\Eq_=}$ asserts that $\phi$ appropriately interacts with the degeneracy $\Eq_=$ and similarly for the analogous components $\phi_{\Eq_\parallel}$, $\phi_{\cm_\top}$, $\phi_{\cm_\bot}$.
We define $\forall_n \, \F(2) \, \overline{Q}$ to be the 2-relation with underlying tuple of relations 
\[\big(\forall_n \F(1) \, \overline{Q^{0\top}}, \forall_n \, \F(1) \, \overline{Q^{1\top}}, \forall_n \, \F(1) \, \overline{Q^{0\bot}}, \forall_n \, \F(1) \, \overline{Q^{1\bot}}\big)\]
mapping $\big(((f_0,f_1,f_2),(g_0,g_1,g_2),(h_0,h_1,h_2),(l_0,l_1,l_2)),((\phi_0,\ldots),(\phi_1,\ldots),(\phi_2,$\linebreak
$\ldots),(\phi_3,\ldots))\big)$ to the proposition
\begin{align*} 
& \Pi_{Q : \mathsf{2R}_0} \F(2)\, (\overline{Q} ,Q) \, \big((f_0 \, Q^{0\top}_\mathsf{d}, g_0 \, Q^{0\top}_\mathsf{c}, h_0 \, Q^{0\bot}_\mathsf{d}, l_0 \, Q^{0 \bot}_\mathsf{c}), \\
& \;\;\;\;\;\;\;\;\;\;\;\;\;\;\;\;\;\;\;\;\;\;\;\;\;\;\;\;\;\;~ (\phi_0 \, Q^{0\top}, \phi_1 \, Q^{1\top}, \phi_2 \, Q^{0\bot}, \phi_3 \, Q^{1\bot})\big)
\end{align*}
\end{example}
Finally, unlike the frameworks
~\cite{dr04,param_johann,gfs16,mr92,rr94,jac99}, our definition of a
parametric model recognizes the above proof-relevant model:

\begin{theorem}[Proof-relevant model]
The family of adjoints defined in Example~\ref{ex:pf-rel} makes
$\mathsf{F}_{\mathit{2D}}$ into a $\lambda 2$-fibration, and hence into a parametric model of
System F over $\mathsf{R}_\mathit{PREY}$.
\end{theorem}

\begin{proof}[Proof sketch]
Faithfulness follows because having $\eta(0),\eta(1)$ fixed, there is a unique way to define $\eta(2)$: since $\eta$ has to respect the degeneracy $\Eq_=$, we must have
\[ \eta(2) \circ_{\rel(2)} \varepsilon^=_\F = \varepsilon^=_\G \circ_{\rel(2)} \Eq_=(\eta(1)) \]
where $\varepsilon^=_\F$, $\varepsilon^=_\G$ are the natural isomorphisms witnessing the fact that $\F$, $\G$ by assumption preserve $\Eq_=$ (we could have used any of the other functors $\Eq_\parallel,\cm_\top,\cm_\bot$ as well). This gives at most one possible value for $\eta(2)$. Fullness follows since the triple $\{\eta(l)\}_{l \in \{0,1,2\}}$ with $\eta(2)$ as given above indeed respects all face maps, degeneracies, and connections (in fact it is only necessary to check the respecting of face maps since the predicates at level 2 are proof-irrelevant). Finally, essential surjectivity follows from the fact that the reflexive graph functor $(\F(0),\F(1))$ is isomorphic to the reflexive graph functor $(\F(0),\Eq(\F(0)))$ via the reflexive graph natural transformation $(\mathsf{id},\varepsilon_\F)$ (the fact that this transformation is face-map preserving again uses the coherence $\rel(\fm_\star) \circ \varepsilon_\F = \id$). But $(\F(0),\Eq(\F(0)))$ clearly belongs to the image since it can be extended \emph{e.g.}, to the triple $\big(\F(0),\Eq(\F(0)),\Eq_=(\Eq(\F(0)))\big)$.         
\end{proof}

\section{Discussion}
We can now be more specific about how our approach compares to the external approaches in \cite{gfs16,jac99,mr92,rr94}, all of which are based on a reflexive graph of $\lambda 2$-fibrations. The definition in \cite{gfs16} appears to be too restrictive: it requires a comprehension structure that, \emph{e.g.}, the $\lambda 2$-fibration corresponding to Reynolds' model does not admit. In addition, none of these frameworks seem to recognize the $\lambda$-fibration corresponding to the proof-relevant model as parametric, for the following reason: it is unclear how to define the family of adjoints for
the second fibration (called $r$ in \cite{jac99}) of ``heterogeneous'' reflexive graph functors in a way that is compatible with the adjoint structure on the original $\lambda 2$-fibration. This is because unlike in the proof-irrelevant case, the definition of $\forall_n \, \F(1)$ now has conditions such as the one witnessed by $\phi_=$ which are only meaningful for ``homogeneous'' reflexive graph functors, \emph{i.e.}, those where the domain and codomain of $\F(1)(\overline{R})$ are given by the \emph{same} functor $\F(1)$, albeit applied to different arguments ($\overline{R_\mathsf{d}}$ vs. $\overline{R}_\mathsf{c}$). Our definition does not rely on or require two compatible adjoint structures, which is why we are indeed able to recognize the proof-relevant model as parametric.

We indicate three directions for future work. Readers interested in \emph{applications of parametricity} will notice that we do not require conditions such as (op)cartesianness or fullness of certain
maps or well-pointedness of certain categories. This follows the spirit of \cite{jac99}, where the notion of \emph{parametricity} pertains to the suitable interaction with (what we call) face maps and
degeneracies. Specific \emph{applications} such as establishing the Graph Lemma and the existence of initial algebras are left for another occasion. Readers fond of \emph{type theory} might wonder about
possible models expressed in the \emph{intensional} version of dependent type theory. Although currently there are no well-known models for which the latter would be the right choice of meta-theory, that might change with more research into higher notions of parametricity. Finally, readers familiar with \emph{cubical sets} no doubt recognized the structure of sets, relations, and 2-relations with face maps,
degeneracies, and connections from the last section as the first few levels of the cubical hierarchy, and wonder whether one can formulate the analogous notion of $2,3, \ldots$-parametricity using this hierarchy. We conjecture the answer to be a \emph{YES!} and plan to pursue this question in future work.

\vspace*{0.05in}

\noindent
{\bf Acknowledgments}
This research is supported by NSF awards 1420175 and 1545197. We thank Steve Awodey and Peter Dybjer for helpful discussions. We also thank the anonymous referee who independently suggested the formulation of Reynolds' model in Example~\ref{ex:rey1}, which appeared in our earlier preprint \cite{preprint}.

\bibliographystyle{plain}
\bibliography{references}

\end{document}